\documentclass[]{article}
\usepackage[english]{babel}

\usepackage{amsfonts, amssymb}
\usepackage{amsmath, amsthm}
\usepackage{dsfont}
\usepackage{mathrsfs}
\usepackage{yfonts}

\usepackage{cancel}
\usepackage{graphicx}
\usepackage{physics}
\usepackage{mathtools}
\usepackage[hidelinks]{hyperref}
\usepackage{subcaption}

\newtheorem{theorem}{Theorem}

\newtheorem{lemma}[theorem]{Lemma}

\newtheorem{prop}[theorem]{Proposition}

\theoremstyle{definition}
\newtheorem{definition}[theorem]{Definition}
\newtheorem*{remark}{Remark}

\usepackage{multicol}
\usepackage[title]{appendix}

\usepackage{cite}

\usepackage{pdfpages}

\usepackage[usenames,dvipsnames]{pstricks}
\usepackage{pstricks-add}
\usepackage{epsfig}
\usepackage{pst-grad} 
\usepackage{pst-plot} 
\usepackage[space]{grffile} 
\usepackage{etoolbox} 

\usepackage{authblk}

\date{February 23, 2026}

\makeatletter
\patchcmd\Gread@eps{\@inputcheck#1 }{\@inputcheck"#1"\relax}{}{}
\makeatother

\DeclareFontEncoding{LS1}{}{}
\DeclareFontSubstitution{LS1}{stix}{m}{n}
\DeclareMathAlphabet{\mathscr}{LS1}{stixscr}{m}{n}
\makeatother

\newcommand{\txtscr}[1]{  \text{\usefont{LS1}{stixscr}{m}{n}#1}}

\newcommand{\bigGroup}{\mathcal{G}}
\newcommand{\elementofbigGroup}{\txtscr{g}}

\newcommand{\ep}[1]{\mathrm{e}^{#1}}
\newcommand{\iu}{\mathrm{i}}

\title{Local measurements and the entanglement transition in quantum spin chains}
\author[1]{Sven Bachmann} 
\author[1]{Mahsa Rahnama} 
\author[2,3]{Gabrielle Tournaire}

\affil[1]{\textit{\small Department of Mathematics, The University of British Columbia, Vancouver, Canada}}

\affil[2]{\textit{ \small Department of Physics and Astronomy, The University of British Columbia, Vancouver, Canada}}

\affil[3]{\textit{\small Stewart Blusson Quantum Matter Institute, The University of British Columbia, Vancouver, Canada}}

\begin{document}

\maketitle

\begin{abstract}
    We consider the transition between short-range entangled (SRE) and long-range entangled states of infinite quantum spin chains which is induced by local measurements. Specifically, we assume that the initial state is in a non-trivial symmetry-protected topological phase of a compact local symmetry group $\mathcal{G}$. We show that the on-site measurements of an appropriate local charge on intervals of increasing lengths transform the initial SRE state into a family of states with increasingly long-range correlations. In particular, the post-measurement states cannot be uniformly short-range entangled. In the case where the initial state is obtained from a product state using a quantum cellular automaton, we construct the infinite-volume post-measurement state and exhibit almost local observables that are maximally correlated. 
\end{abstract}

\section{Introduction}
Since the seminal work~\cite{hastings2005quasiadiabatic} the classification of topological orders of quantum spin systems has relied on the equivalence relation given by local unitaries. Because topologically ordered states are gapped, the physical motivation behind this definition comes from the adiabatic principle: If two states lie at the endpoint of a path of uniformly gapped ground states, then they can be deformed to each other adiabatically and this dynamical flow can be realized locally~\cite{bachmann2018adiabatic}: at zeroth order, this is the quasi-adiabatic continuation, also known as the local spectral flow~\cite{bachmann2012automorphic}. For mathematical purposes, the explicit construction of a path of Hamiltonians is not essential (although it is physically important), and one is left with an equivalence between states given by what is now commonly referred to as locally generated automorphisms (LGA), see for example~\cite{kapustin2022local}: A dynamics over a finite time, which is generated by a time-dependent, sufficiently local interaction.

The key property of LGAs is that they satisfy a Lieb-Robinson bound: propagation is at most ballistic, up to errors that are small outside the linear `sound cone'~\cite{Lieb_Robinson_bound_1972,Nachtergaele_2019}. This ensures in particular that long-range entanglement in the initial state cannot be created or deleted by the dynamics, which is physically the defining property of an equivalence class of topological orders. One may wish to enforce an even stronger locality by considering quantum cellular automata (QCA)~\cite{schumacher2004reversibleQCA} rather than LGAs, which are dynamics with strictly ballistic propagation. Examples thereof, which is a natural family from a quantum computational point of view, are finite depth quantum circuits (FDQC)~\cite{hastings2005quasiadiabatic, Chen_2010_LREvsSRE}. Not surprisingly, FDQCs are good approximations of LGAs by Trotter-type results~\cite{ChildsSuTranWiebeZhu2021,BachmannLange2022}. Reciprocally, in one spatial dimension at least, classifying states using QCAs is equivalent to classifying them using approximately locality preserving automorphisms~\cite{ranard2022converse}, which in turn are equivalent to Lieb-Robinson type automorphisms~\cite{bachmann2025Locality}.

The observation that local measurements, unlike the local unitaries described above, are able to create long-range entanglement dates back to the invention of the one-way quantum computer~\cite{raussendorf2001one} (but see~\cite{ogata2026noteinvariantsmixedstatetopologicalorder} for instances where quantum channels reduce the topological order). In the topological implementation of this measurement-based version of quantum computation~\cite{raussendorf2007fault,tournaire20243d}, toric code eigenstates are created by local measurements of a short-range entangled cluster state. In fact, a similar creation of long-range entanglement by local measurements extends to a large class of topological codes in two dimensions~\cite{bolt2016foliated}, see also~\cite{hoke2023measurement} for experimental challenges. This surprising observation brings up a natural question~\cite{skinner2019measurement}: what is the topological classification of states under an equivalence defined by both LGAs and local measurements? A first complete discussion including many conjectures can be found in~\cite{tantivasadakarn2023hierarchy, tantivasadakarn_long-range_2024}, where measurements are shown to be able to promote short-range entangled, but symmetry-protected topological states to full long-range entangled orders. There, the relationship is explored in detail in the one-dimensional case of matrix product states with an Abelian gauge group. This `entanglement transition' induced by projection is further explored in~\cite{cheng2024universal}.

In this paper, we analyse the entanglement transition induced by local measurements on a symmetry-protected topological (SPT) state of infinite quantum spin chains, for a general compact group $\bigGroup$. Our proof exhibits how local measurements transform the hidden string order that is typical of SPT states (as  exhibited in the AKLT chain~\cite{AKLT}, see also~\cite{den1989preroughening,KennedyTasaki1992} and~\cite{SOandSPT_PerezGarcia_2008})  into a classical long-range order. In this sense, local measurements perform an analog of a Kennedy-Tasaki transformation~\cite{KennedyTasaki1992CMP,Oshikawa1992}, although it is a distinct operation. One of the interests of our construction, which does not need the matrix product machinery, is that it emphasizes the role of the general cohomological SPT index~\cite{Chen2011Classification,Pollmann_2012_SPT,Ogata2021SplitProperty}. Here, we shall use the recent generalization~\cite{carvalho_classification_2024}. It therefore opens up the analysis to more general cases, including higher dimensions.

Specifically, we consider a pure short-range entangled (SRE) state $\omega$ of an infinite quantum spin chain which is invariant under the local action of a compact symmetry group $\mathcal{G}$. Assuming that $\omega$ is in a non-trivial SPT phase, namely is cannot be smoothly mapped to a product state while preserving both locality and symmetry, we first show that it exhibits string order, see Proposition~\ref{prop: string order} for details: There is at least one group element $\tilde g$ such that for any two sites $i<j$, there are observables $W^i, W^j$ of norm one, localized around $i,j$, and such that $\omega(W^i) = 0 = \omega(W^j)$ while
\begin{equation*}
    \vert \omega(W^i U_{\tilde g}^{[i,j)} W^j)\vert = 1
\end{equation*}
where $U_{\tilde g}^{[i,j)}$ is the unitary implementing the symmetry on the interval $[i,j)$.

With this, we can prove the main theorem, Theorem~\ref{theorem: entanglement growth}. For any Abelian group $G\ni\tilde g$, for example the group generated by $\tilde g$, there is a good notion of measuring the $G$-charge at any site $i$ of the chain: This is described by a projection $P^{(i)}_{q}$ associated with the measurement outcome labelled by~$q$. If $\omega_{n,\mathbf{q}}$ denotes the state after measurements on the interval $[-n,n]$, we prove that for $-n\leq i<j\leq n$, we again have $\omega_{n,\mathbf{q}}(W_n^i) = 0 = \omega_{n,\mathbf{q}}(W_n^j)$ but
\begin{equation*}
    \vert \omega_{n,\mathbf{q}}(W_n^{i} W_n^j) - \omega_{n,\mathbf{q}}(W_n^i)\omega_{n,\mathbf{q}}(W_n^j)\vert = 1,
\end{equation*}
without the need of the connecting string $U_{\tilde g}^{[i,j)}$. As a consequence, the family of states cannot be uniformly short-range entangled. A slightly stronger version of the theorem, Theorem~\ref{theorem: long range correlations with QCA}, can be proved whenever the initial SRE state $\omega$ arises from the action of a QCA on a pure state: In this case, the operators $W_n^{i}$ are eventually independent of $n$ and the statement extends from the sequence $(\omega_n)_{n\in\mathbb{N}}$ to any limiting state $\nu$ after the full chain has been measured.

Summarizing, local measurements trade the hidden string order of the SPT phase for a manifest long-range order.

We wish to emphasize here that the nature of the symmetry group $\bigGroup$ is not the crucial element that ensures an entanglement transition. Rather, we exhibit conditions on the initial SPT \emph{state}, and identify the appropriate \emph{measurements} to be carried out, so that an entanglement transition is induced. 

We shall describe our main results in Section~\ref{sec:Intro & Results}, where we present the simple case of the one-dimensional cluster state in parallel to the general situation. In Section~\ref{sec:SPT}, we consider symmetry-protected topological phases and recall the construction of the projective representation characterizing the phases. This is intimately related to string order. With this, we turn in Section~\ref{sec: correlations in the measured state} to the measured state and prove our first main result, Theorem~\ref{theorem: entanglement growth}. Finally, we show in Section~\ref{sec: long range correlation for QCA} that in the case where the entangler is a QCA, appropriately blocked local measurements yield an infinite-volume post-measurement state which is long-range entangled. 

\section{Setup and main result}\label{sec:Intro & Results}
We briefly recall the elementary mathematical setup of infinite quantum spin chains, in particular what we mean by local measurements. We also present a few useful lemmas following these definitions. 

\subsection{Spin chain $C^*$-algebra and states}

To each site $j\in \mathbb{Z}$ of the spin chain we associate the algebra of observables $\mathcal{A}_j\cong M_{d_j}(\mathbb{C})$ and shall assume that $d_j\leq d$ for all $j\in \mathbb{Z}$. When equipped with the usual norm, it is a $C^*$-algebra. Let $\mathcal{P}_{\rm fin}$ be the collection of finite subsets of $\mathbb{Z}$. Then for each $S\in \mathcal{P}_{\rm fin}$, we define $\mathcal{A}_S=\otimes_{j\in S}\mathcal{A}_j$ and the algebra of observables on the infinite spin chain is defined as
\begin{equation*}
    \mathcal{A}= \overline{\bigcup_{S\in\mathcal{P}_{\rm fin}}\mathcal{A}_S},
\end{equation*}
where overline denotes the closure with respect to the operator norm. Note that for any subset $X\subset \mathbb{Z}$, the algebra $\mathcal{A}_X$ defined similarly is a subalgebra of $\mathcal{A}$, and $\mathcal{A}=\mathcal{A}_X\otimes\mathcal{A}_{X^c}$. Throughout this paper, we will often use the notion of \emph{half chain}. That is, we will consider operators with support contained only on the right of some given site. Let $L_j=(-\infty,j)$ and $R_j=[j,+\infty)$, the left and right half chains at site $j\in \mathbb{Z}$. For all such $j$, the algebra $\mathcal{A}$ can be written as $\mathcal{A}=\mathcal{A}_{L_j}\otimes\mathcal{A}_{R_j}$. The set of unitary elements of these algebras shall be denoted $\mathcal{U}$, respectively $\mathcal{U}_X$ for $X\subset\mathbb{Z}$.

\paragraph{Conditional expectations.} Following \cite{Nachtergaele_2019}, for every finite \(X \subset \mathbb{Z}\), we extend the notion of the partial trace to the infinite volume setting. There exists a conditional expectation
\begin{equation*}
   \Pi_X : \mathcal{A} \longrightarrow \mathcal{A},
\end{equation*}
characterized by the properties:
\begin{enumerate}
   \item \(\Pi_X\) is unital, completely positive, and contractive;
   \item \(\Pi_X\) is idempotent with range \(\mathcal{A}_X\), and \(\Pi_X|_{\mathcal{A}_X} = \mathrm{id}_\mathcal{A}\);
   \item \(\Pi_X\) is \(\mathcal{A}_X\)-bimodular:
   \[
      \Pi_X(B_1 A B_2) = B_1 \Pi_X(A) B_2,
      \qquad B_1,B_2 \in \mathcal{A}_X,\; A \in \mathcal{A}.
   \]
\end{enumerate}
Equivalently, for each \(A \in \mathcal{A}\) there exists a unique \(A_X \in \mathcal{A}_X\) such that
\begin{equation*}%\label{eq:PiX-structure}
   \Pi_X(A) = A_X \otimes \mathbb{I}_{X^c}.
\end{equation*}

\paragraph{Almost local observables.} We now fix a class of decay functions used throughout the quasi-locality estimates. Let $\mathcal{F}$ be the class of bounded non-increasing positive functions $f:\mathbb{N} \to \mathbb{R}_+$  that vanish faster than any power, i.e.\ $\lim_{r\to\infty} r^p f(r) = 0$ for all $p>0$.
For a finite subset $X \subset \mathbb Z$ and $r \ge 0$, we define the
$r$-neighborhood of $X$ by
\begin{equation*}%\label{eq:neighbourhood}
X_r := \{ j \in \mathbb Z : \mathrm{dist}(j,X) \le r \}.
\end{equation*}
Let now $I$ be a finite interval. Given a function $f \in \mathcal F$, we say that $A$ is
\emph{$f$-close to $I$} if for all $r \ge 0$,
\begin{equation}\label{eq:f-close}
\| A - \Pi_{I_r}(A) \| \le f(r)\Vert A\Vert.
\end{equation}
We call \emph{$f$-local} an operator $A\in \mathcal{A}$ that is $f$-close to some finite interval $I$. We say that $A$ is \emph{almost local} if it is $f$-local for some $f\in\mathcal{F}$.

\paragraph{States.} A state $\omega$ on the spin chain is a normalized positive linear functional on the algebra $\mathcal{A}$. The set of states on $\mathcal{A}$ is a convex set which is compact with respect to the weak-* topology. A pure state is an extremal point of the convex set of states. A product state is a state that satisfies
\begin{equation*}
    \omega(A_iA_j)=\omega(A_i)\omega(A_j),
\end{equation*}
for all $A_i \in \mathcal{A}_i$ and $ A_j\in \mathcal{A}_j$, whenever $i\neq j$.

\subsection{Short-Range Entanglement}
While the notion of short-range entanglement introduced in~\cite{Chen_2010_LREvsSRE} and briefly described in the introduction is based on LGAs (or slight variants thereof), we shall here adopt a slightly more general notion, which however encapsulates the same locality property, namely that of a \emph{factorizable} automorphism.

\subsubsection{Local automorphisms}

For any $U\in\mathcal{U}$, we define the $^*$-automorphism of $\mathcal{A}$
 \begin{equation*}
     \mathrm{Ad}[U](A)=UAU^* 
 \end{equation*}
where $A\in \mathcal{A}$.

\begin{definition}\label{def:split automorphism}
Let $\alpha$ be a $^*$-automorphism on $\mathcal{A}$.
We say that $\alpha$ satisfies the \emph{factorization property} if for each site $j\in \mathbb{Z}$, there exist automorphisms
\[
\alpha^{L_j} \in \mathrm{Aut}(\mathcal A_{L_j}), 
\qquad 
\alpha^{R_j} \in \mathrm{Aut}(\mathcal A_{R_j}),
\]
and a unitary $U^j \in \mathcal U$ such that
\begin{equation}\label{eq: split-automorphism}
\alpha = (\alpha^{L_j} \otimes \alpha^{R_j}) \circ \mathrm{Ad}[U^j].
\end{equation}
We further require that there is $f \in \mathcal F$ such that $U^j$ is $f$-close to the singleton $\{j\}$ for all $j$.
\end{definition}
\noindent Note that the decay function $f$ characterizing the almost locality of $U^j$ is independent of the site $j$.

\begin{definition}\label{def: alpha almost local}
    We say that a $^\ast$-automorphism $\alpha$ is $f$-almost local if
for every finite interval $I \subset \mathbb Z$ and every
$A \in \mathcal A_I$, the observable $\alpha(A)$ is $f$-close to
$I$ in the sense of \eqref{eq:f-close}.
\end{definition}

The next lemma shows that the factorization property of an automorphism $\alpha$ implies an almost-locality estimate for $\alpha$.
\begin{lemma}\label{lemma: split implies quasi local}
    Let $\alpha$ be a $^*$-automorphism on $\mathcal{A}$ satisfying the factorization property, and let $f\in \mathcal{F}$ be the corresponding decay function. Then there is $f_\alpha\in \mathcal{F}$ depending only on $f$ such that $\alpha$ is $f_\alpha$-almost local.
\end{lemma}

\begin{proof}
Let $I=[i,j]$ with $i<j$. Let $r>0$ and consider the factorization property at both sites $i'=i-r$, $j'=j+r$. Thus, $\alpha$ satisfies
\begin{equation*}
\alpha^{L_{i'}}\otimes\alpha^{R_{i'}}\circ\mathrm{Ad}[U^{i'}]
=\alpha
=\alpha^{L_{j'}}\otimes\alpha^{R_{j'}}\circ\mathrm{Ad}[U^{j'}].
\end{equation*}
For $A\in \mathcal{A}_I$, the fact that $U^{i'}$ is $f$-localized around $i'$ implies it almost commutes with $A$, namely
\begin{equation*}
    \|\mathrm{Ad}[U^{i'}](A)-A\|\leq 4\|A\|f(r),
\end{equation*}
and the same holds for $i'\leftrightarrow j'$. Therefore, and since $A$ is supported on $R^{i'}\cap L^{j'}$, 
\begin{align}
    \|\alpha(A)-\alpha^{R_{i'}}(A)\|&\leq 4\|A\|f(r), \label{eq: alpha is almost on Ri}\\
    \|\alpha(A)-\alpha^{L_{j'}}(A)\|&\leq 4\|A\|f(r), \label{eq: alpha is almost on Lj}
\end{align}
and in turn
\begin{equation}\label{eq:LiRi}
    \|\alpha^{R_{i'}}(A)-\alpha^{L_{j'}}(A)\|\leq8\|A\|f(r).
\end{equation}

Furthermore, since $\mathcal{A}_{R_{i'}}$ is a quasi-local algebra, there is a finite subset $\Lambda\subset R_{i'}$ such that
\begin{equation}\label{eq: quasi locality of alpha Ri }
    \|\alpha^{R_{i'}}(A)-\Pi_\Lambda(\alpha^{R_{i'}}(A))\|\leq \|A\|f(r). 
\end{equation}
Note that, by possibly enlarging $\Lambda$, we may assume that $[i',j']\subset\Lambda$. Using that $\Pi_\Lambda$ is a norm-one conditional expectation, we obtain from~(\ref{eq:LiRi}) that
\begin{equation}\label{eq: Pi_Lambda or Ri and Lj is f close}
    \|\Pi_\Lambda(\alpha^{R_{i'}}(A))-\Pi_\Lambda(\alpha^{L_{j'}}(A))\|\leq8\|A\|f(r)
\end{equation}
and we note that $\Pi_\Lambda(\alpha^{L_{j'}}(A))=\Pi_{I_r}(\alpha^{L_{j'}}(A))$ since $\Lambda\cap L_{j'} \subset [i',j']$ and $[i',j']\subset\Lambda$. We conclude from~\eqref{eq: alpha is almost on Ri}, \eqref{eq: alpha is almost on Lj}, \eqref{eq: quasi locality of alpha Ri } and~\eqref{eq: Pi_Lambda or Ri and Lj is f close} that
\begin{equation*}
    \|\alpha(A)-\Pi_{I_r}(\alpha(A))\| \leq 17 \|A\|f(r),
\end{equation*}
concluding the proof, with $f_\alpha = 17 f$.
\end{proof}

\begin{remark}
    The inverse of $\alpha$ is also a $^*$-automorphism satisfying the factorization property \eqref{eq: split-automorphism} as
    \begin{equation*}
        \alpha^{-1}=(\alpha^{L_j})^{-1}\otimes(\alpha^{R_j})^{-1}\circ\mathrm{Ad}[\alpha^{L_j}\otimes\alpha^{R_j}(U^{j*})], \quad \forall j\in \mathbb{Z}. 
    \end{equation*}
    Since $\alpha^{L_j}\otimes\alpha^{R_j}$ is almost local, it maps almost local observables to almost local observables. Therefore, $\alpha^{L_j}\otimes\alpha^{R_j}(U^{j*})$ is $f'$-close to $\{j\}$ with $f'$ depending only on $f_{\alpha}$. In the following, we denote $f_\alpha\in \mathcal{F}$ a function such that $\alpha$ and $\alpha^{-1}$ are $f_\alpha$-almost-local.
\end{remark}

Although Lemma~\ref{lemma: split implies quasi local} gives a Lieb-Robinson bound type of estimate on the spread of operators for factorizable automorphisms, we will also consider the case where the spread is strict.

\begin{definition}\label{def: QCA}
    Let $\alpha$ be a $^*$-automorphism on $\mathcal{A}$.
We say that $\alpha$ is a \emph{Quantum Cellular Automaton} (QCA) if there exists a finite radius \(N\ge1\) such that for every finite interval $I\in \mathcal{P}_{\rm fin}$ and every \(A\in\mathcal{A}_I\), we have that $\alpha(A)\in\mathcal{A}_{I_N}$.
\end{definition}

\subsubsection{Short-range entangled states}
We finally come to the definition of short-range entangled state, made here using $^*$-automorphisms that satisfy the factorization property.

\begin{definition}\label{def: short range entangled states}
A state $\omega$ on $\mathcal A$ is
called \emph{short-range entangled} (SRE) if there exist $f\in\mathcal{F}$, a product state $\omega_0$ on $\mathcal A$ and a strongly continuous
one-parameter family of $^*$-automorphisms $\{\alpha_s\}_{s\in[0,1]}$ such that for all $s\in[0,1]$, $\alpha_s$ satisfies the factorization property for $f$ (Definition~\ref{def:split automorphism})
with $\alpha_0 = \mathrm{id}$ and
\begin{equation*}
    \omega = \omega_0 \circ \alpha_1.
\end{equation*}
\end{definition}
\noindent A state $\omega$ which is not SRE is called long-range entangled (LRE).

An example of a pure short-range entangled state, which will be used many times in this paper, is the one-dimensional cluster state \cite{briegel2001persistent}. The on-site Hilbert space is $\mathbb{C}^2$ and the on-site algebra $\mathcal{A}_i$ is that of $2\times2$ matrices. The reference product state is $\omega_+(A_i) = \langle +_i \vert A_i\, +_i\rangle$, where $\ket{+_i}=\frac{1}{\sqrt{2}}(\ket{0_i}+\ket{1_i})$ and $A_i\in\mathcal{A}_i$. The group of automorphisms $\alpha_s$ is generated by the (formal) Hamiltonian $\sum_{j} h_j$, where
%h_j=\frac{1}{\}\ket{1_j1_{j+1}}\bra{1_j1_{j+1}}$
$h_j = \frac{\pi}{4}(\mathbb{I}_j - Z_j)(\mathbb{I}_{j+1} - Z_{j+1}) = \pi\ket{1_j1_{j+1}}\bra{1_j1_{j+1}}$ are nearest-neighbour interactions ($Z_j, X_j$ denote the Pauli matrices at the site $j$) \cite{Raussendorf_2012_QCbyLocalMeas}. This is a commuting Hamiltonian and so the automorphism 
\begin{equation}\label{eq: split automorphism cluster state}
    \alpha_s(A)=\mathrm{Ad}[U_{\rm{cluster}}^X(s)](A), \quad \forall A\in\mathcal{A}_X, ~X\in\mathcal{P}_{\rm{fin}},
\end{equation}
where $U_{\rm{cluster}}^X(s)= \prod_{j:\{j,j+1\}\cap X\neq \emptyset}\mathrm{e}^{\mathrm{i}sh_j }$, is in fact a QCA. It is in particular a factorizable automorphism and $\omega_{\rm{cluster}} = \omega_+\circ\alpha_1$ is SRE.

\subsection{Symmetry Protected Topological phases}

In this section we provide a standard definition of Symmetry Protected Topological phases on the one-dimensional quantum spin chain, see \cite{Schuch2011Classifying_SPT_MPS,Chen2011Classification, ogata2021classificationsymmetryprotectedtopological, carvalho_classification_2024}. First, we introduce the notion of global group symmetries that have an on-site action. We then consider pure short-range entangled states that are invariant under these symmetry actions, and classify them further into topological phases.

\subsubsection{On-site symmetries}
\label{sec: group symmetries}
Let $\bigGroup$ be a compact group, and let for any $j\in\mathbb{Z}$ the map $\elementofbigGroup\rightarrow U^{(j)}_\elementofbigGroup$ be a unitary linear representation. We write
\begin{equation*}%\label{eq: def U^X_g}
    U^X_\elementofbigGroup\coloneqq \bigotimes_{j\in X}U^{(j)}_\elementofbigGroup.
\end{equation*}
Then $\beta_\elementofbigGroup^{(j)}(A) = \text{Ad}[U^{(j)}_\elementofbigGroup](A)$ defines a group action on $\mathcal{A}_{\{j\}}$ similarly for
\begin{equation*}
    \beta^X_\elementofbigGroup(A) = \text{Ad}[U^X_\elementofbigGroup](A)
\end{equation*}
on $\mathcal{A}_X$ for any $X\in \mathcal{P}_{\rm fin}$. This is extended to all $A\in\mathcal{A}$ by linearity and density. Similarly, we can define $\beta_\elementofbigGroup^\Gamma$ on $\mathcal{A}_\Gamma$ for all finite or infinite subsets $\Gamma \subseteq \mathbb{Z}$. 

\begin{definition}
    A state $\omega$ is $\bigGroup$\emph{-invariant} or \emph{symmetric under $\bigGroup$ action} if
    \begin{equation*}
        \omega \circ \beta_\elementofbigGroup= \omega
    \end{equation*}
     for all $\elementofbigGroup\in \bigGroup$.
\end{definition}

\subsubsection{State equivalence and topological phases}

Let $\{\beta_\elementofbigGroup\}_{\elementofbigGroup\in \bigGroup}$ be as above. A family $\{\alpha_s\}_{s\in [0,1]}$ of $^*$-automorphisms is said to be $\bigGroup$-equivariant if
\begin{equation*}
    \alpha_s\circ \beta_\elementofbigGroup=\beta_\elementofbigGroup\circ \alpha_s,
\end{equation*}
for all $\elementofbigGroup\in \bigGroup, ~s\in [0,1]$.

\begin{definition}
    Two $\bigGroup$-invariant pure SRE states $\omega,\nu$ on $\mathcal{A}$ are \emph{$\bigGroup$-equivalent} if there exist $f\in\mathcal{F}$ and a $\bigGroup$-equivariant strongly continuous one-parameter family of $^*$-automorphisms $\{\alpha_s\}_{s\in[0,1]}$ such that for all $s\in[0,1]$, $\alpha_s$ satisfies the factorization property for $f$, and
\begin{equation*}
    \omega=\nu\circ\alpha_1. 
\end{equation*}
This defines equivalence classes on $\bigGroup$-invariant pure SRE states that are called \emph{Symmetry Protected Topological ($\bigGroup$-SPT) phases.}
\end{definition}

The cluster state defined in \eqref{eq: split automorphism cluster state} is a good example of an SPT state for the symmetry group $\mathcal{G} = \mathbb{Z}_2\times\mathbb{Z}_2$. Here, we block sites in pairs, each new site containing an `odd' and an `even' spin. Then the on-site symmetry is given by the local unitary representation
\begin{equation*}
    (0,0)\mapsto \mathbb{I}_{\mathrm{o}}\otimes\mathbb{I}_{\mathrm{e}}, \quad(0,1) \mapsto \mathbb{I}_{\mathrm{o}}\otimes X_{\mathrm{e}}, \quad (1,0) \mapsto X_{\mathrm{o}}\otimes\mathbb{I}_{\mathrm{e}}, \quad (1,1)\mapsto X_{\mathrm{o}}\otimes X_{\mathrm{e}}. 
\end{equation*}
where the subscripts $\mathrm{o}$ and $\mathrm{e}$ stand for the odd and even subsite, respectively. The symmetry does not leave the entangling interaction invariant, and so $\alpha_s\circ\beta_{\elementofbigGroup}\neq \beta_{\elementofbigGroup}\circ\alpha_s$ for all $0<s<1$ and $\elementofbigGroup=(0,1)$ or $(1,0)$. Clearly, the product state $\omega_+$ is symmetric under the $\bigGroup$ action since every spin is initially in the eigenstate of $X$. Since 
\begin{equation}\label{eq: stabilizers}
    U_\mathrm{cluster}^I(1) X_j = Z_{j-1}X_jZ_{j+1}U_\mathrm{cluster}^I(1)
\end{equation}
whenever $\{j-1,j,j+1\}\subset I$ and $Z^2 = \mathbb{I}$, we have that $\alpha_1\circ\beta_{\elementofbigGroup}=\beta_{\elementofbigGroup}\circ\alpha_1$,
see also \cite{Raussendorf_2012_QCbyLocalMeas}. We conclude that
\begin{equation*}
   \omega_{\mathrm{cluster}}\circ\beta_{\elementofbigGroup} 
   = \omega_+\circ\alpha_1\circ\beta_{\elementofbigGroup}
   = \omega_+\circ\beta_{\elementofbigGroup}\circ\alpha_1
   = \omega_{\mathrm{cluster}}
\end{equation*}
therefore the state $\omega_{\mathrm{cluster}}$ is a $\bigGroup$-invariant SRE state. 

The classification problem for SPT phases in one dimension has been completely solved, see again~\cite{Chen2011Classification, ogata2021classificationsymmetryprotectedtopological,carvalho_classification_2024} as well as the closely related~\cite{Kapustin_2021_classification_bosonicSPT}. While we will need the details later, see Section~\ref{subsec: projective representation and index}, we here only recall the result. To any pure $\bigGroup$-invariant SRE state $\omega$, one associates a projective representation of $\bigGroup$, which is characterized by a $2$-cocycle $\mu_\omega:\bigGroup\times\bigGroup\to U(1)$. The $\bigGroup$-SPT phases are then in one-to-one correspondence with the second cohomology group of $\bigGroup$, and so the cohomology class of the cocycle $\mu_\omega$ can be used as an index that completely classifies the $\bigGroup$-SPT phases.

A state $\omega$ is in the trivial phase if and only if $[\mu_\omega] = [1]$, namely the representation can be made into a linear one by a trivial redefinition of the phases of the unitaries. In the following and unlike in previous works, we shall neither assume that $\bigGroup$ is Abelian, nor that it factorizes as a direct product. For any pair of commuting $\elementofbigGroup,\elementofbigGroup'\in\bigGroup$, we define
\begin{equation}\label{eq: def sigma omega}
\sigma_\omega(\elementofbigGroup,\elementofbigGroup')=\frac{\mu_\omega(\elementofbigGroup,\elementofbigGroup')}{\mu_\omega(\elementofbigGroup',\elementofbigGroup)}.
\end{equation}
Since the two elements commute, $\rho(\elementofbigGroup\elementofbigGroup')=\rho(\elementofbigGroup'\elementofbigGroup)$ for any $\rho:\bigGroup\to U(1)$, and so $\sigma_\omega(\elementofbigGroup,\elementofbigGroup')$ does not depend on the representative of the cohomology class of $\mu_\omega$. In particular, if $\omega$ is in a trivial phase, then $\sigma_\omega(\elementofbigGroup,\elementofbigGroup')=1$ for every commuting pair. Following~\cite{tantivasadakarn_long-range_2024}, we shall assume that this is not the case. By the previous remark, this property is an invariant of the $\bigGroup$-SPT phase.

\begin{definition}\label{def: non trivial mixed SPT}
    Let $\bigGroup$ be a compact group. A $\bigGroup$-invariant pure SRE state $\omega$ is \emph{non-degenerate} if there are two commuting elements $\tilde{g},\tilde{h}\in \bigGroup$ such that
    \begin{equation}\label{eq: sigma mixed}
        \sigma_\omega(\tilde{g},\tilde{h})\neq 1.
    \end{equation}
\end{definition}

\noindent For a finite group $\bigGroup$, an element $\elementofbigGroup$ is called $\mu$-\emph{regular} if $\mu(\elementofbigGroup,\elementofbigGroup')=\mu(\elementofbigGroup',\elementofbigGroup)$ for every $\elementofbigGroup'$ that commutes with $\elementofbigGroup$. In that case, Definition~\ref{def: non trivial mixed SPT} states precisely that $\tilde g$ fails to be $\mu_\omega$-regular.

If $\bigGroup=G\times H$ is a direct product with $\tilde g\in G$ and $\tilde h\in H$, then $\tilde g$ and $\tilde h$ commute automatically and the definition reduces to the \emph{non-trivial mixed} property of~\cite{tantivasadakarn_long-range_2024}. We point out however that the present setting is strictly more general, and this in two distinct ways. First, $\bigGroup$ need not be Abelian: only the subgroup whose charge is measured will have to be, see Section~\ref{subsec:local measurements} below. Second, even for Abelian $\bigGroup$ the closed subgroup generated by $\tilde g$ and $\tilde h$ need not split as a direct product of $\overline{\langle\tilde g\rangle}$ and $\overline{\langle\tilde h\rangle}$.

In Section~\ref{subsec: projective representation and index}, we will pick up our running example of the one-dimensional cluster state and prove that $\omega_{\mathrm{cluster}}$ is non-degenerate for the group $\mathbb{Z}_2\times\mathbb{Z}_2$ in the sense of the definition above.

\subsection{Local measurements}
\label{subsec:local measurements}

Let again $\bigGroup$ be a compact group and let $\omega$ be a pure, $\bigGroup$-invariant, non-degenerate SRE state, with $\tilde g,\tilde h$ as in Definition~\ref{def: non trivial mixed SPT}. We shall measure the charge of a subgroup $G$ of $\bigGroup$ subject to the following two standing assumptions:
\begin{enumerate}
    \item[(M1)] $G$ is a closed Abelian subgroup of $\bigGroup$ with $\tilde g\in G$;
    \item[(M2)] $\tilde h\in Z_{\bigGroup}(G) = \{h\in\bigGroup:hg=gh\text{ for all } g\in G\}$, the centralizer of $G$ in $\bigGroup$.
\end{enumerate}

The canonical choice is the smallest one, $G=\overline{\langle\tilde g\rangle}$, the closed subgroup topologically generated by $\tilde g$: then (M1) is immediate and (M2) reduces to $\tilde h\in Z_{\bigGroup}(\tilde g)$, which is part of Definition~\ref{def: non trivial mixed SPT}. Note that (M2) cannot be weakened to the requirement that $\tilde h$ merely normalize $G$.

In the rest of the paper, we use the notation convention that general elements of $\bigGroup$ are written in script, $\elementofbigGroup\in\bigGroup$, while elements of the Abelian subgroup $G$ whose charge is measured are written in roman, $g\in G$, and so are elements of $Z_{\bigGroup}(\tilde g)$.

By (M1), $G$ is compact, and we let $\mu$ be its Haar measure. Then $G\ni g\mapsto U_g$ is a finite-dimensional representation of $G$. Because $G$ is Abelian, the family $\{U_g:g\in G\}$ is commuting and hence can be simultaneously diagonalized, namely
\begin{equation*}
    U_g = \sum_{q}\ep{-\iu \lambda_q(g)} P_q,
\end{equation*}
where $P_q$ are the spectral projectors, $\ep{-\iu \lambda_q(g)}$ are the eigenvalues, and $q$ labels the characters $g\mapsto \overline{\chi_q}(g) = \ep{-\iu \lambda_q(g)}$. If $\mu$ denotes the Haar measure of $G$, then
\begin{equation*}
    \int_G \chi_{q'}(g) U_g {\rm d}\mu(g) = \sum_{q} \int_G \chi_{q'}(g)\overline{\chi_{q}}(g) {\rm d} \mu(g) P_{q} = P_{q'}
\end{equation*}
by orthogonality of the characters.

\paragraph{Examples.} 
(i) For the one-dimensional cluster state, we pick $G=\mathbb{Z}_2$ with $U_0=\mathbb{I}_{\mathrm{e}}$ and $U_1=X_{\mathrm{e}}$ on the `even' spin of the site. Thus, the projectors are 
\begin{equation*}
    P_q=\frac{\mathbb{I_{\mathrm{e}}}+(-1)^qX_{\mathrm{e}}}{2},
\end{equation*}
for $q=0$ or $1$.\\
\noindent (ii) Let $G = U(1)$. Denoting $g = \ep{\iu \nu}$, we let $U_\nu$ be the two-dimensional representation given by $U_\nu = \ep{-\iu \nu Z}$. Then $\lambda_1(\nu) = \nu$ while $\lambda_2(\nu) = -\nu$ and 
\begin{equation*}
    P_{1,2} = \frac{1}{2\pi}\int_0^{2\pi}\ep{\pm \iu \nu}\ep{-\iu \nu Z}{\rm d}\nu
\end{equation*}
are the spectral projectors of $Z$.

This motivates the following definition.
\begin{definition}\label{def: local measurments}
    Let $\bigGroup$ be a compact group, let $G$ be a closed Abelian subgroup of~$\bigGroup$, and let $j\in\mathbb{Z}$. The \emph{measurement of $G$ at site $j$} is given by the family
\begin{equation}\label{eq:Projectors}
    P^{(j)}_q = \int_G \chi_{q}(g) U^{(j)}_g {\rm d} \mu(g),\qquad 1\leq q\leq d_j
\end{equation}
of projectors in $\mathcal{A}_{\{j\}}$.
\end{definition}

We further define measurements $P_{n,\mathbf{q}} = \bigotimes_{j\in [-n,n]}P^{(j)}_{q_j}$ that span all the sites between $-n$ and $n$. Every site in this interval undergoes exactly one measurement, with outcome given by the vector $(\lambda_{q_{-n}},\ldots,\lambda_{q_n})$. 

Now, if $\omega$ is a state on $\mathcal{A}$, then
\begin{equation}\label{eq: def of w_n}
    \omega_{n,\mathbf{q}}(A)=\frac{1}{\mathcal{N}_{n,\mathbf{q}}}\,\omega\big(P_{n,\mathbf{q}} A P_{n,\mathbf{q}}\big)
\end{equation}
is the state obtained after measurements on the sites in $[-n,n]$ have yielded outcomes labelled by $\mathbf{q}$. Here, the normalization constant is given by
\begin{equation*}
    \mathcal{N}_{n,\mathbf{q}}=\omega(P_{n,\mathbf{q}})
\end{equation*}
and it is the probability to obtain the measurement outcome labelled by $\mathbf{q}$. It is therefore natural to assume that $\mathcal{N}_{n,\mathbf{q}}\neq 0$ (and hence \eqref{eq: def of w_n} is well-defined) because the complementary event occurs with probability zero.

\begin{remark}
While the symmetry group $\bigGroup$ protecting the phase is arbitrary, it is essential that the measured subgroup $G$ be Abelian. The non-Abelian generalization of~\eqref{eq:Projectors} would not yield the collapse of the string operator to a scalar described in Section~\ref{subsec: computation of correlation after measurements}. The measurement would in general leave a higher-dimensional space at each measured site on which the string operator would still act non-trivially.
\end{remark}

\subsection{Main results}

We first show that a non-degenerate, $\bigGroup$-invariant pure SRE state as in Definition~\ref{def: non trivial mixed SPT} displays string order.
\begin{prop}\label{prop: string order}
Let $\bigGroup$ be a compact group. Let $\omega$ be a non-degenerate, $\bigGroup$-invariant pure SRE state with commuting group elements $\tilde{g},\tilde{h}$ satisfying~\eqref{eq: sigma mixed}. Then, there is $f\in\mathcal{F}$ such that for each site $j\in \mathbb{Z}$ there exists a unitary $W^j_{\tilde{g}}\in\mathcal{A}$ that is $f$-close to $j$ such that
    \begin{equation*}
        \omega(W^j_{\tilde{g}})=0
    \end{equation*}
    and for all pairs of sites $i<j$, 
    \begin{equation*}
        \vert \omega(W^i_{\tilde{g}}U^{[i,j)}_{\tilde{g}}W^j_{\tilde{g}})\vert = 1.
    \end{equation*}
\end{prop}
\noindent As we already mentioned, this string order in SPT phases plays a crucial role in the following. In fact, it is the string order in the pre-measured state $\omega$ that is responsible for the long-range order in the post-measurement state.

In the example of the cluster state, the string order parameters are given by 
\begin{equation}\label{eq: cluster state string parameter}
    Z_{\mathrm{o}}^{i} X_{\mathrm{e}}^{[i,j)} Z^j_{\mathrm{o}} \quad \mathrm{and}
    \quad Z_{\mathrm{e}}^{i-1} X_{\mathrm{o}}^{[i,j)} Z^{j-1}_{\mathrm{e}}.
\end{equation}
This follows from the observation that $Z_{\mathrm{o}}^{i} X_{\mathrm{e}}^{[i,j)} Z^j_{\mathrm{o}} = \alpha^{-1}_{\rm cluster}(X_{\mathrm{e}}^{[i,j)})$, see~\eqref{eq: stabilizers}, which yields
\begin{equation*}
    \omega_{\rm cluster}(Z_{\mathrm{o}}^{i} X_{\mathrm{e}}^{[i,j)} Z^j_{\mathrm{o}}) = \omega_+(X_{\mathrm{e}}^{[i,j)}) = 1,
\end{equation*}
and similarly for the other one.

In the general case, our main result shows that, although each post-measurement state
$\omega_n$ \eqref{eq: def of w_n} is short-range entangled (see Lemma~\ref{lemma: post meas SRE}), the
almost-locality bounds necessarily deteriorate with the size of the
measurement region and cannot be chosen uniformly in $n$.

\begin{definition}
    A family $(\alpha_n)_{n\in\mathbb{N}}$ of factorizable $^*$-automorphisms of $\mathcal{A}$ is \emph{uniformly almost-local} if there exists $f\in\mathcal{F}$ such that every $\alpha_n$ is $f$-almost local as in Definition~\ref{def: alpha almost local}.
\end{definition}

\begin{theorem}\label{theorem: entanglement growth}
Let $\bigGroup$ be a compact group. Let $\omega$ be a non-degenerate, $\bigGroup$-invariant pure SRE state with commuting group elements $\tilde{g},\tilde{h}$ satisfying~\eqref{eq: sigma mixed}. Let $G$ be a subgroup of $\bigGroup$ satisfying (M1) and (M2). Let $(\omega_n)_{n\in\mathbb N}$ be the family of post-measurement states defined in \eqref{eq: def of w_n}.

Then there is no uniformly almost-local family $(\alpha_n)_{n\in\mathbb N}$ of factorizable $^*$-automorphisms such that $\omega_n=\omega_0\circ\alpha_n$ for some product state $\omega_0$ and for all $n\in \mathbb{N}$.  
\end{theorem}

In fact, for each $n$, we consider any factorizable $^*$-automorphism $\alpha_n$ such that $\omega_n=\omega_0\circ\alpha_n$ and the corresponding decay function  $f_{\alpha_n}\in \mathcal{F}$ as in Lemma~\ref{lemma: split implies quasi local}. We show in Section~\ref{sec: correlations in the measured state} that these decay functions cannot be uniformly bounded in $\mathcal{F}$, namely
\begin{equation*}
   \sup_{n\in\mathbb N} f_{\alpha_n} \notin \mathcal F.
\end{equation*}

A stronger result holds when assuming that the entangler $\alpha$ associated with $\omega$ strictly preserves the locality of the observables, namely it is a QCA in the sense of Definition~\ref{def: QCA}. By blocking finitely many sites for measurements, arbitrarily long-range correlations can be obtained. In this case, we further assume that
\begin{equation}\label{eq: M3}
    \sigma_\omega(\tilde g, g) = 1\qquad\text{for all } g\in G,
\end{equation}
with $\sigma_\omega$ defined in~\eqref{eq: def sigma omega}. We shall prove below, see Lemma~\ref{lemma: sigma bicharacter}, that \eqref{eq: M3} holds automatically for the canonical choice $G=\overline{\langle\tilde g\rangle}$. For any other Abelian $G\ni \tilde g$ it can be recovered by passing to the closed subgroup $G\to \ker\sigma_\omega(\tilde g,\cdot)\cap G$, which still contains $\tilde g$ and still satisfies (M1) and (M2). It is therefore a choice of measurement rather than a restriction on the state.

We will show that for each pair of blocks of sites $i$ and $j$, there are two local operators $\widetilde{W}^{i}_{\tilde{g}}$ and $\widetilde{W}^{j}_{\tilde{g}}$, such that for all $n$ with $i$ and $j$ in the measurement region of the state $\tilde\omega_n$, 
    \begin{equation}\label{eq: LR correlation QCA section 2}
        |\tilde\omega_n(\widetilde{W}^{i}_{\tilde{g}}\widetilde{W}^{j}_{\tilde{g}})-\tilde\omega_n(\widetilde{W}^{i}_{\tilde{g}})\tilde\omega_n(\widetilde{W}^{j}_{\tilde{g}})|=1. 
    \end{equation}
Since the set of states on $\mathcal{A}$ is compact with respect to the weak-$*$ topology, the sequence $(\tilde\omega_n)_{n\in\mathbb{N}}$ admits a weakly converging subsequence, along which correlation functions such as~\eqref{eq: LR correlation QCA section 2} converge. In turn, this implies the emergence of long-range entanglement in any limiting state. We discuss the blocking protocol and prove the following theorem in Section~\ref{sec: long range correlation for QCA}.

\begin{theorem}\label{theorem: long range correlations with QCA}
    Let $\omega = \omega_0\circ\alpha$ be as in Theorem~\ref{theorem: entanglement growth} with the additional assumptions that $\alpha$ is a QCA and that~\eqref{eq: M3} holds. 
    Let $\nu$ be a limiting state of the sequence $(\tilde\omega_n)_{n\in\mathbb{N}}$ of block-measured states. Then $\nu$ is long-range entangled. 
\end{theorem}

We conclude with the example of the cluster state, see also~\cite{briegel2001persistent,tantivasadakarn_long-range_2024}. We pick $\tilde g = (0,1), \tilde h = (1,0)$ and $G = \langle \tilde g\rangle = \{(0,0),(0,1)\}$, corresponding to measuring $X$'s on the `even' sites. The resulting measured state is so that
\begin{equation*}
    \tilde \omega_n(Z_{\rm o}^iZ_{\rm o}^j) = \pm 1,\qquad \tilde \omega_n(Z_{\rm o}^i) = 0 = \tilde \omega_n(Z_{\rm o}^j),
\end{equation*}
where $\pm 1$ is the product of the measurement outcomes between $i$ and $j$, for any $i,j$ in the measured region, see also Section~\ref{sec:QCA post correlations}. On the `odd' sub-lattice of the measured region, this is the highly entangled GHZ state $\vert{\rm GHZ}\rangle = \frac{1}{\sqrt 2}(\vert 1\cdots 1\rangle + \vert 0\cdots 0\rangle)$, up to local unitaries that depend on the measurement outcomes. In the infinite volume limit, we obtain $\nu = \frac{1}{2}(\nu_+ + \nu_-)$, the equal mixture of the two pure ferromagnetic states.

As we will see in Section~\ref{sec: long range correlation for QCA}, Theorem~\ref{theorem: long range correlations with QCA} holds in fact for any, even non-QCA, factorizable automorphism under the additional assumption that 
\begin{equation*}
    \omega(P_{n,\textbf{q} }\mathrm{Ad}[W^j_{\tilde{g}}](A)P_{n,\textbf{q} })= \omega\circ \mathrm{Ad}[W^j_{\tilde{g}}](P_{n,\textbf{q} }AP_{n,\textbf{q} }), \quad \forall A \in \mathcal{A},
\end{equation*}
where $W^j_{\tilde{g}}$ is part of the string order parameter in Proposition~\ref{prop: string order}. In fact, this commutation relation implies that the states $\omega, ~\omega_n$ have the same string order, but the unitaries $U^{[i,j)}_{\tilde{g}}$ act trivially after the measurements. 

Finally, we point out that while the measured states $\omega_n$ are all pure, the weak-* limiting state $\nu$ is mixed, as exhibited in the cluster state example. This follows from~(\ref{eq: LR correlation QCA section 2}) using the general property that pure states are clustering, see Section~\ref{Section: nu is mixed}.

\section{On SPT phases}\label{sec:SPT}

We start by exploring the consequence of Definition~\ref{def: short range entangled states}, which relies on the minimal locality assumption of a factorizable automorphism. We then turn to the role of the local symmetry and briefly recall the construction of the SPT index, which arises from the action of symmetry on the half-infinite chain, see~\cite{ogata2021classificationsymmetryprotectedtopological, carvalho_classification_2024}. This will allow us to describe precisely what we mean by string order and prove Proposition~\ref{prop: string order}.

\subsection{Factorization property and short-range entanglement}

Definition~\ref{def: short range entangled states} generalizes the idea that short-range entangled states are obtained from trivial
product states by finite-time quasi-local evolutions, as first formulated in \cite{Chen_2010_LREvsSRE}. Indeed, short-range entanglement is often defined using locally generated automorphisms (LGAs) arising from a time–dependent
local interaction \(K(s)\), see \cite{Kapustin_2021_classification_bosonicSPT,ogata2021classificationsymmetryprotectedtopological,bachmann_classification_2024,carvalho_classification_2024,bachmann_many-body_2024}. From the point of view of this paper, LGAs
form an important subclass
of automorphisms and satisfy the factorization property \eqref{eq: split-automorphism}. Indeed, if \(\alpha_s^K\) denotes the dynamics generated by
the time-dependent interaction \(K(s)\), then it satisfies the factorization property
\begin{equation*}
\alpha_s^K= (\alpha_s^{K_L}\otimes\alpha^{K_R}_s)\circ\mathrm{Ad}[U_s]
\end{equation*}
by the Lieb-Robinson bound~\cite{Lieb_Robinson_bound_1972,Nachtergaele_2006ExponentialClustering}. Here, $K_L$ and $K_R$ denote the interaction strictly on the left and on the right, respectively, and the almost local unitary $U_s$ solves the differential equation
\begin{equation*}
    \frac{d}{ds}\mathrm{Ad}[U_s](A)=\mathrm{Ad}[U_s]\big(\mathrm{i}[\alpha_s^{K_L}\otimes\alpha^{K_R}_s(K_M(s)), A]\big),
\end{equation*}
with $K_M$ the interactions between the left and right chains, see~\cite{bachmann2017local}. Lemma~\ref{lemma: split implies quasi local} then implies that for any observable $A$ localized on a finite interval, $\alpha_s^K(A)$ is almost local.

\subsection{Decay of correlations}

While the SRE condition translates the physical intuition that the state $\omega$ does not exhibit long-range entanglement, it implies in particular short-range correlations. This holds again by locality of the dynamics relating $\omega$ to the product state~\cite{Bravyi_2006_Topo_order_is_LRE, Hastings_2006ExponentialDecay, Nachtergaele_2006Propagation_Correlation}. A similar bound continues to hold with our Definition~\ref{def: short range entangled states} which does not explicitly depend on a local Hamiltonian.

\begin{lemma}\label{lemma: decay of correlation}
Let $\omega$ be an SRE state $\omega = \omega_0\circ\alpha$. Then for any intervals $I,J$ with $l= \mathrm{dist}(I,J)$ and any $A \in \mathcal{A}_I$, $B \in \mathcal{A}_J$, 
\begin{equation*}
    |\omega(AB) - \omega(A)\omega(B)|
   \leq 4\|A\| \|B\| f_\alpha\left(\frac{l}{3}\right)
\end{equation*}
 where $f_\alpha$ is given by Lemma~\ref{lemma: split implies quasi local}.
\end{lemma}

\begin{proof}
By Lemma~\ref{lemma: split implies quasi local},
\begin{equation}\label{bound on R}
    \|\alpha(A) - \Pi_{I_r}(\alpha(A))\|
      \le \|A\| f_\alpha(r), 
\qquad 
   \|\alpha(B) - \Pi_{J_r}(\alpha(B))\|
      \le \|B\| f_\alpha(r).
\end{equation}
We define $A_r := \Pi_{I_r}(\alpha(A)), B_r := \Pi_{J_r}(\alpha(B))$. For any $r\leq \frac{l}{3}$, $\mathrm{dist}(I_r,J_r) = l - 2r \geq l/3$, so $A_r$ and $B_r$ have disjoint supports, and thus
\[
   \omega_0(A_r B_r) = \omega_0(A_r)\, \omega_0(B_r),
\]
because $\omega_0$ is a product state. Next, let
\begin{equation*}
   \alpha(A) = A_r + R_r(A), 
   \qquad
   \alpha(B) = B_r + R_r(B),
\end{equation*}
%where $R_r(A)= \alpha(A) - A_r$. 
Since $\omega(AB)
   = \omega_0(\alpha(A)\alpha(B))$, 
\begin{equation*}
   \omega(AB)
   = \omega_0(A_r B_r) 
     + \omega_0(R_r(A) \alpha(B)) 
     + \omega_0(A_r R_r(B))
     %+ \omega_0(R_r(A)R_r(B)),
\end{equation*}
and hence
\begin{equation}\label{corr 1}
    \vert \omega(AB) - \omega_0(A_r B_r) \vert
    \leq  2\|A\|\|B\| f_\alpha(r),
\end{equation}
by the locality bound~(\ref{bound on R}). By the same bound,
\begin{equation}
    |\omega_0(A_r) - \omega(A)| 
    = |\omega_0(R_r(A))|
    \le \|A\| f_\alpha(r),
\end{equation}
and similarly with $A\leftrightarrow B$. Thus
\begin{equation}\label{corr 2}
   |\omega_0(A_r)\omega_0(B_r) - \omega(A)\omega(B)|
      \le 2\|A\|\|B\| f_\alpha(r).
\end{equation}
Combining (\ref{corr 1},\ref{corr 2}), we conclude that
\begin{equation*}
    |\omega(AB) - \omega(A)\omega(B)| \leq 4\|A\|\|B\| f_\alpha(r).
\end{equation*}
The claim follows since $f_\alpha$ is decreasing and the bound holds for all $r\leq \frac{l}{3}$.
\end{proof}

\subsection{Action of the half-chain symmetry on a short-range entangled state}

 For any site $j$ and $\elementofbigGroup\in\bigGroup$, consider the automorphism $\beta^{R_j}_\elementofbigGroup$ and $\beta^{L_j}_\elementofbigGroup$. On the left from the cut $j$, the automorphism $\beta^{R_j}_\elementofbigGroup$ acts as the identity. While far on the right, it behaves like the full symmetry. Since the state is symmetric and short-range entangled, $\beta^{R_j}_\elementofbigGroup$ acts as an almost local operator on $\omega$, see~\cite{Kapustin_2021_classification_bosonicSPT, bachmann_classification_2024}.

\begin{prop}
\label{prop: existence of W^R_g}
Let $\omega$ be a $\bigGroup$-invariant SRE state. Then for each $j\in\mathbb{Z}$ and each $\elementofbigGroup \in \mathcal{G}$, there exists a unitary $W^{R_j}_\elementofbigGroup \in \mathcal{A}$ such that
\begin{equation}
\omega \circ \beta^{R_j}_\elementofbigGroup  = \omega \circ \mathrm{Ad}[W^{R_j}_\elementofbigGroup].
\label{eq: existence of W^R_g}
\end{equation}
Furthermore, $W^{R_j}_\elementofbigGroup$ is $f_W$-close to the singleton $\{j\}$, where $f_W\in \mathcal{F}$ depends only on~$f_\alpha$ and neither on $j$ nor on $\elementofbigGroup\in\bigGroup$.
\end{prop}

\begin{remark}
Of course, the same result holds with $R_j$ replaced with $L_j$, yielding an almost-local operator $W^{L_j}_\elementofbigGroup$. As the proof below shows, the function $f_W$ can be explicitly written as $f_W(r)=F(f_\alpha)(r)$, where $F:\mathcal{F}\rightarrow\mathcal{F}$ is independent of $\alpha$. Moreover, the map $F$ is increasing, a property that we will use in Section~\ref{sec:Proof of Main 1}.
\end{remark}

\begin{proof}
Without loss of generality, let us consider $j=0$ and we note $R\coloneqq R_0$. Denoting $\omega = \omega_0\circ \alpha$ and using its $\bigGroup$-invariance, 
\begin{equation*}
    \omega_0 \circ \alpha \circ \beta^{R}_\elementofbigGroup \circ \alpha^{-1}=\omega\circ\beta_{\elementofbigGroup^{-1}}\circ\beta_\elementofbigGroup^{R}\circ\alpha^{-1}
    =\omega\circ\beta^L_{\elementofbigGroup^{-1}}\circ\alpha^{-1},
\end{equation*}
because $\beta_{\elementofbigGroup^{-1}}$ is a group action. Let $A\in\mathcal{A}_{\rm loc}$ be supported in $(r,+\infty)$. Using the factorization property of $\alpha^{-1}$, see Lemma~\ref{lemma: split implies quasi local}, we have that $\alpha^{-1}(A)$ is almost localized on $R$, namely there is $\tilde{A} \in \mathcal{A}_{R}$ such that:
\begin{equation*}
    \|\alpha^{-1}(A) - \tilde{A}\| \leq f_\alpha(r) \|A\|.
\end{equation*}
Since $\tilde A$ is invariant under $\beta_{\elementofbigGroup^{-1}}^L$, we conclude that $\alpha^{-1}(A)$ is almost invariant
\begin{equation*}
   \|\beta_{\elementofbigGroup^{-1}}^L(\alpha^{-1}(A)) - \alpha^{-1}(A)\| \leq 2f_\alpha(r)\|A\|. 
\end{equation*}
Hence,
\[
\left|\omega_0 \circ \alpha \circ \beta_\elementofbigGroup^R \circ \alpha^{-1}(A) - \omega_0(A) \right| 
\leq \|\alpha\circ\beta_\elementofbigGroup^R \circ \alpha^{-1}(A) - A\| 
\leq 2f_\alpha(r)\|A\|
\]
for all $A \in \mathcal{A}_{(r,+\infty)}$ with finite support.

A similar argument (which does not require the symmetry of $\omega$) yields the same bound for all $A \in \mathcal{A}_{(-\infty,-r)}$ with finite support. Therefore,
\[
|\omega_0\circ\alpha\circ\beta^{R}_\elementofbigGroup \circ \alpha^{-1}(A) - \omega_0(A)| \leq 2f_\alpha(r)\|A\|.
\]

This implies, see e.g.\ Lemma~A.1 of~\cite{bachmann_many-body_2024}, that there is $f\in\mathcal{F}$, depending only on $f_\alpha$, such that
\begin{equation*}
    |\omega_0 \circ \alpha \circ \beta^{R}_\elementofbigGroup \circ \alpha^{-1}(A) - \omega_0(A)| \leq f(r)\|A\|.
\end{equation*}
In turn, by Proposition A.2 of~\cite{bachmann_many-body_2024}, there exists a unitary $\tilde{W}^{R}_{\elementofbigGroup}\in \mathcal{A}$ that is almost localized at $\{0\}$ such that
\begin{equation*}
    \omega_0 \circ \alpha \circ \beta^{R}_\elementofbigGroup \circ \alpha^{-1} = \omega_0 \circ \mathrm{Ad}[\tilde{W}^{R}_{\elementofbigGroup}].
\end{equation*}
 The almost locality of $\tilde{W}^R_\elementofbigGroup$ depends again only on $f_\alpha$. Pulling this back to $\omega$ yields
\begin{equation*}
\omega \circ \beta^{R_i}_\elementofbigGroup(B) 
= \omega \circ \alpha^{-1} \circ \mathrm{Ad}[\tilde{W}^{R}_{\elementofbigGroup}](\alpha(B))
= \omega \circ \mathrm{Ad}[\alpha^{-1}(\tilde{W}^{R}_{\elementofbigGroup})](B),
\end{equation*}
so we can set
\begin{equation*}
    W^{R}_{\elementofbigGroup} = \alpha^{-1}(\tilde{W}^{R}_{\elementofbigGroup}).
\end{equation*}
The unitary $W^{R}_{\elementofbigGroup}$ is again almost localized since $\alpha^{-1}$ is an almost-local automorphism, and the corresponding localization function $f_W$ depends only on~$f_\alpha$. In fact, using the estimates of the reference above, one finds that $W^{R}_{\elementofbigGroup}$ is $f_W$ almost localized with
\begin{equation}\label{eq: F of f almost locality of WR}
   f_W(r)= f_\alpha(\frac{r}{2})+ 36\sum^\infty_{l=\frac{r}{2}} \sqrt{ 2f_\alpha(l) +2\sum_{k=l}^\infty  \sqrt{df_\alpha(k)}},
\end{equation}
concluding the proof.
\end{proof}

\subsection{Projective representation of the group $\bigGroup$}
\label{subsec: projective representation and index}

Let again $\omega$ be a $\bigGroup$-invariant pure SRE state. Following the classification results in the literature \cite{ogata2021classificationsymmetryprotectedtopological, Kapustin_2021_classification_bosonicSPT, carvalho_classification_2024}, we show the existence of a projective representation of the group $\bigGroup$ in the GNS representation $(\mathcal{H},\pi,\Omega)$ of $\omega$, and discuss its relation with the local operators ${W^{R_j}_\elementofbigGroup}$. Since $\omega$ is pure, $\pi$ is irreducible.

\begin{prop}
\label{prop: existence of VRg + properties}
   Let $(\mathcal{H}, \pi, \Omega)$ be a GNS representation of the $\bigGroup$-invariant pure SRE state $\omega$. Let $j\in\mathbb{Z}$. There exists a unitary projective representation $\elementofbigGroup\mapsto V_{\elementofbigGroup}^{R_j}$ of $\bigGroup$ such that
    \begin{align}
    \pi \circ \beta_{\elementofbigGroup}^{R_j} &=   \mathrm{Ad}[V_{\elementofbigGroup}^{R_j}]\circ \pi,  \label{eq: V^R_g properties}\\
         V_{\elementofbigGroup}^{R_j*} \Omega &= \pi(W_{\elementofbigGroup}^{R_j*}) \Omega \label{eq: V^R_g properties2}.
    \end{align}
\end{prop}
\noindent In particular, we have that
\begin{equation}
        V_{\elementofbigGroup}^{R_j*} \left(\pi \circ \beta_{\elementofbigGroup}^{R_j}\right) (A) \Omega = \pi\left(A W_{\elementofbigGroup}^{R_j*}\right) \Omega 
        \label{eq: definition of V^R_g}
    \end{equation}
    for all $A\in\mathcal{A}$.
    
\begin{proof}
  Let $\elementofbigGroup\in\bigGroup$. The state $\omega_\elementofbigGroup \coloneqq \omega\circ\beta^{R_j}_\elementofbigGroup$ has two GNS representations on $\mathcal H$, namely 
  \begin{equation*}
     (\mathcal{H}, \pi \circ \beta_{\elementofbigGroup}^{R_j}, \Omega)\quad\text{and}\quad
     (\mathcal{H}, \pi, \pi(W_{\elementofbigGroup}^{R_j*}) \Omega),
  \end{equation*}
where the second one is by Proposition~\ref{prop: existence of W^R_g}. Uniqueness of the GNS representation implies that there exists a unitary intertwiner $V_{\elementofbigGroup}^{R_j}: \mathcal{H} \to \mathcal{H}$ satisfying (\ref{eq: V^R_g properties},\ref{eq: V^R_g properties2}).

It remains to prove that $\elementofbigGroup\to V_{\elementofbigGroup}^{R_j}$ is in general a projective representation. Since $\beta_{\elementofbigGroup}^{R_j} \circ \beta_{\elementofbigGroup'}^{R_j} =\beta_{\elementofbigGroup \elementofbigGroup'}^{R_j} $, we have by~\eqref{eq: V^R_g properties} that 
   \begin{equation*}
       \text{Ad}[V_{\elementofbigGroup\elementofbigGroup'}^{R_j}]\circ \pi
       = \pi\circ\beta^{R_j}_\elementofbigGroup\circ\beta^{R_j}_{\elementofbigGroup'}
       =  \text{Ad}[V_{\elementofbigGroup}^{R_j}V_{\elementofbigGroup'}^{R_j}]\circ \pi.
   \end{equation*}
   Since $\omega$ is pure, the representation $\pi$ is irreducible and hence
\begin{equation*}
    V_{\elementofbigGroup}^{R_j} V_{\elementofbigGroup'}^{R_j} = \mu_\omega^{(j)}(\elementofbigGroup,\elementofbigGroup') V_{\elementofbigGroup\elementofbigGroup'}^{R_j}
\end{equation*}
for some $\mu_\omega^{(j)}(\elementofbigGroup,\elementofbigGroup')\in U(1)$.
\end{proof}

\begin{remark}
Importantly, the cocycle $\mu_\omega^{(j)}(\elementofbigGroup,\elementofbigGroup')$ can also be computed using the almost local $W_{\elementofbigGroup}^{R_j}\in\mathcal{A}$. Indeed,
\begin{align*}
    \mu_\omega^{(j)}(\elementofbigGroup,\elementofbigGroup')&= \langle \Omega, V_{\elementofbigGroup}^{R_j} V_{\elementofbigGroup'}^{R_j} V_{\elementofbigGroup\elementofbigGroup'}^{R_j*}\Omega \rangle\\
    &= \langle \Omega, \pi(W_{\elementofbigGroup}^{R_j}) V_{\elementofbigGroup'}^{R_j} \pi(W_{\elementofbigGroup\elementofbigGroup'}^{R_j*})\Omega \rangle \\
    &= \langle \Omega, \pi(W_{\elementofbigGroup'}^{R_j}\beta_{(\elementofbigGroup')^{-1}}^{R_j}(W_{\elementofbigGroup}^{R_j}) )\pi(W_{\elementofbigGroup\elementofbigGroup'}^{R_j*})\Omega \rangle
\end{align*}
by~(\ref{eq: V^R_g properties2},\ref{eq: definition of V^R_g}), and so
\begin{equation}\label{eq: index as function of W^R}
    \mu_\omega^{(j)}(\elementofbigGroup,\elementofbigGroup') 
    = \omega \left( W_{\elementofbigGroup'}^{R_j}\beta_{(\elementofbigGroup')^{-1}}^{R_j}(W_{\elementofbigGroup}^{R_j})W_{\elementofbigGroup\elementofbigGroup'}^{R_j*}\right).
\end{equation}
\end{remark}

Moving the cut from site $i$ to site $j$ amounts to acting with a local unitary and yields a transformation of the cocycle by a coboundary, namely, $\mu_\omega^{(i)}$ and $\mu_\omega^{(j)}$ belong to the same cohomology class.  It follows that if $\elementofbigGroup,\elementofbigGroup'\in\bigGroup$ commute, the definition~\eqref{eq: def sigma omega} of the index
\begin{equation*}
\sigma_\omega(\elementofbigGroup,\elementofbigGroup') = \frac{\mu_\omega^{(j)}(\elementofbigGroup,\elementofbigGroup')}{\mu_\omega^{(j)}(\elementofbigGroup',\elementofbigGroup)},
\end{equation*}
is independent of the site $j$. We further note that it is in fact nothing other than a commutator (temporarily dropping the superscript $R_j$)
\begin{equation*}
\sigma_\omega(\elementofbigGroup,\elementofbigGroup') = V_\elementofbigGroup V_{\elementofbigGroup'} V_{\elementofbigGroup\elementofbigGroup'}^*V_{\elementofbigGroup'\elementofbigGroup} V_\elementofbigGroup^* V_{\elementofbigGroup'}^* = V_\elementofbigGroup V_{\elementofbigGroup'} V_\elementofbigGroup^* V_{\elementofbigGroup'}^*
\end{equation*}
because $\elementofbigGroup$ and $\elementofbigGroup'$ commute.

\begin{lemma}\label{lemma: sigma bicharacter}
Let $G$ be a closed Abelian subgroup of $\bigGroup$. Then the restriction of $\sigma_\omega$ to $G\times G$ satisfies $\sigma_\omega(a,a) = 1$ and $ \sigma_\omega(a,b) = \overline{\sigma_\omega(b,a)}$, and
\begin{equation*}
    \sigma_\omega(a,bc) = \sigma_\omega(a,b)\sigma_\omega(a,c)
\end{equation*}
for all $a,b,c\in G$.

Moreover, for any $a\in G$, $\ker\sigma_\omega(a,\cdot)$ is a closed subgroup of $G$ containing $a$. If $G=\overline{\langle a\rangle}$ is generated by a single element, then $\sigma_\omega\upharpoonright_{G\times G}\equiv 1$.
\end{lemma}  
\begin{proof}
Dropping the superscript $R_j$ again, all elements of $G$ commute, so that $\sigma_\omega(a,b)=V_aV_bV_a^*V_b^*$ by the above and the first two properties are immediate. Moreover,
\begin{align*}
    \sigma_\omega(a,bc) &= V_aV_{bc}V_a^*V_{bc}^* = V_aV_bV_cV_a^*V_c^*V_b^*
    = \sigma_\omega(a,c)\, V_aV_bV_a^*V_b^* \\
    &= \sigma_\omega(a,c)\sigma_\omega(a,b),
\end{align*}
where we used $V_cV_a^*V_c^* = (V_cV_aV_c^*)^* = \overline{\sigma_\omega(c,a)}V_a^* = \sigma_\omega(a,c)V_a^*$. It follows that $\sigma_\omega(a,\cdot)$ is a character. It is measurable because $\mu_\omega$ is, and a measurable character of a compact group is continuous by Weil's theorem, so that $\ker\sigma_\omega(a,\cdot)$ is closed. It contains $a$ because $\sigma_\omega(a,a)=1$. If $G=\overline{\langle a\rangle}$, then $\sigma_\omega(a,\cdot)\equiv 1$ on $\langle a\rangle$ by multiplicativity, hence on $G$ by continuity, and conjugate symmetry yields $\sigma_\omega\upharpoonright_{G\times G}\equiv 1$.
\end{proof}

We turn again to the example of the cluster state. Recall from~\eqref{eq: cluster state string parameter} that $W^{R_j}_{(0,1)}=Z^{(j)}_\mathrm{o}$, $W^{R_j}_{(1,0)}=Z^{{(j-1)}}_\mathrm{e}$ and $W^{R_j}_{(1,1)}=Z^{(j-1)}_\mathrm{e} Z^{(j)}_\mathrm{o}$. Thus, from~\eqref{eq: index as function of W^R},
\begin{equation*}
    \mu^{(j)}((0,1),(1,0))=\omega_{\mathrm{cluster}}(Z^{(j-1)}_\mathrm{e} (X^{(j)}_\mathrm{o} Z^{(j)}_\mathrm{o}X^{(j)}_\mathrm{o})  Z^{(j)}_\mathrm{o} Z^{(j-1)}_\mathrm{e})=-1
\end{equation*}
because $X^{(j)}_\mathrm{o}$ and  $Z^{(j)}_\mathrm{o}$ anti-commute. On the other hand, 
\begin{equation*}
    \mu^{(j)}((1,0),(0,1))=\omega_{\mathrm{cluster}}(Z^{(j)}_\mathrm{o}Z^{(j-1)}_\mathrm{e}Z^{(j)}_\mathrm{o} Z^{(j-1)}_\mathrm{e})=1
\end{equation*}
since $\beta_{(0,1)}^{R_j}$ acts trivially on the site $(j-1)$. Hence, 
\begin{equation}\label{eq: index for cluster state}
    \sigma_{\mathrm{cluster}}((0,1),(1,0))=-1,
\end{equation}
and so $\omega_{\mathrm{cluster}}$ is non-degenerate indeed.

\begin{lemma}\label{lem:translating Vs}
Let $\elementofbigGroup \in \bigGroup$. For any $A \in \mathcal{A}_{L_i}$,
\begin{equation}\label{eq: V and op commutating}
    [\pi(A), V_{\elementofbigGroup}^{R_i}] = 0. 
\end{equation}
There is $\gamma(\elementofbigGroup,i,j)\in[0,2\pi)$ such that
\begin{equation}\label{eq: from V^R_i to V^R_j}
    V^{R_i}_\elementofbigGroup = e^{i\gamma(\elementofbigGroup,i,j)} \pi(U^{[i,j)}_\elementofbigGroup) V^{R_j}_\elementofbigGroup.
\end{equation}
Furthermore,
\begin{equation}\label{eq: V^L and V^R commutation}
    [V^{L_i}_\elementofbigGroup, V^{R_i}_{\elementofbigGroup'}]=0
\end{equation}
for all $\elementofbigGroup, \elementofbigGroup'\in\bigGroup$. Finally, if $\elementofbigGroup\elementofbigGroup' = \elementofbigGroup'\elementofbigGroup$, then
\begin{equation}\label{eq: V and U commutation}
    [\pi(U_{\elementofbigGroup'}^{(k)}), V_{\elementofbigGroup}^{R_i}] = 0
\end{equation}
for all $k \in \mathbb{Z}$.
\end{lemma}

\begin{proof}
Let $A \in \mathcal{A}_{L_i}$ and $B \in \mathcal{A}$. Then by~\eqref{eq: definition of V^R_g},
\begin{align*}
    \pi(A)\pi(B)\pi(W_{\elementofbigGroup}^{R_j*})\Omega
    &=V^{R_i*}_\elementofbigGroup(\pi\circ\beta^{R_i}_\elementofbigGroup)(AB)\Omega \\
    &=V^{R_i*}_\elementofbigGroup\pi(A)(\pi\circ\beta^{R_i}_\elementofbigGroup)(B)\Omega \\
    &=V^{R_i*}_\elementofbigGroup\pi(A)V^{R_i}_\elementofbigGroup\pi(B)\pi(W_{\elementofbigGroup}^{R_j*})\Omega.
\end{align*}
The vector $\Omega$ being cyclic, so is $\pi(W_{\elementofbigGroup}^{R_j*})\Omega$ and we conclude that $\pi(A) = V^{R_i*}_\elementofbigGroup\pi(A)V^{R_i}_\elementofbigGroup$, namely~\eqref{eq: V and op commutating}.

As $\beta^{R_i}_\elementofbigGroup = \text{Ad}[U^{[i,j)}_\elementofbigGroup]\circ \beta^{R_j}_\elementofbigGroup$ for $i<j$, we have that
\begin{align*}
    \text{Ad}[V^{R_i}_\elementofbigGroup]\circ\pi
    &= \pi\circ \beta^{R_i}_\elementofbigGroup \\
    &=\pi\circ\text{Ad}[U^{[i,j)}_\elementofbigGroup]\circ \beta^{R_j}_\elementofbigGroup =\text{Ad}[\pi(U^{[i,j)}_\elementofbigGroup) V^{R_j}_\elementofbigGroup]\circ\pi.
\end{align*}
With this, \eqref{eq: from V^R_i to V^R_j} follows by the irreducibility of the GNS representation.

Since $\alpha$ satisfies the factorization property, the state $\omega$ has a GNS representation of the form
\begin{equation*}
    \begin{split}
    &(\mathcal{H}_{L_i}\otimes\mathcal{H}_{R_i}, \pi^{L_i}\otimes\pi^{R_i},\Omega)\\ 
    &= \left(\mathcal{H}_{L_i}\otimes\mathcal{H}_{R_i}, (\pi^{L_i}_{0} \circ\alpha^{L_i})\otimes(\pi^{R_i}_{0}\circ\alpha^{R_i}) , \pi^{L_i}\otimes\pi^{R_i}(U^{i*})\Omega^{L_i}_{0}\otimes\Omega^{R_i}_{0}\right) 
    \end{split}
\end{equation*}
with $(\mathcal{H}_{\Gamma_i}, \pi^{\Gamma_i}_{0}, \Omega^{\Gamma_i}_{0})$ the GNS representation of the product state $\omega_0^{\Gamma_i}$ restricted to the half chain $\Gamma_i$, and $U^i$ is the factorizing unitary of~\eqref{eq: split-automorphism}. In particular, $\pi^{\Gamma_i}$ is an irreducible representation of $\mathcal{A}_{\Gamma_i}$ in $\mathcal{H}_{\Gamma_i}$.
From Proposition~\ref{prop: existence of VRg + properties}, we have
\begin{equation*}
\pi_{L_i}\otimes(\pi_{R_i}\circ\beta^{R_i}_\elementofbigGroup)= \text{Ad}[V^{R_i}_\elementofbigGroup]\circ\pi_{L_i}\otimes\pi_{R_i}.
\end{equation*}
Moreover, the invariance of $\omega_0^{R_i}$ yields a unitary $v^{R_i}_\elementofbigGroup\in \mathcal{B}(\mathcal{H}_{R_i})$ such that $\pi^{R_i}\circ\beta^{R_i}_\elementofbigGroup=\text{Ad}[v^{R_i}_\elementofbigGroup]\circ\pi_{R_i}$. We conclude that 
\begin{equation*}
    V^{R_i}_\elementofbigGroup
    =r(\elementofbigGroup,i)\,\mathbb{I}_{\mathcal{H}_{L_i}}\otimes v^{R_i}_\elementofbigGroup,
\end{equation*}
and similarly, 
\begin{equation*}
V^{L_i}_{\elementofbigGroup}=l(\elementofbigGroup,i)\, v^{L_i}_{\elementofbigGroup}\otimes\mathbb{I}_{\mathcal{H}_{R_i}}
\end{equation*}
for $r(\elementofbigGroup,i),l(\elementofbigGroup,i)\in\mathbb{C}$ of modulus $1$. Since these operators act on separate tensor factors, they commute. Finally, we note that this commutation relation holds in all GNS representations since they are all unitarily equivalent, yielding~\eqref{eq: V^L and V^R commutation}. 

Finally, let $i,k\in\mathbb{Z}$. Then \eqref{eq: V and U commutation} follows from~\eqref{eq: V and op commutating} if $k\leq i$. If $i<k$, then by \eqref{eq: from V^R_i to V^R_j},
\begin{align*}
    \pi(U_{\elementofbigGroup'}^{(k)}) V_{\elementofbigGroup}^{R_i}
    &= \pi(U_{\elementofbigGroup'}^{(k)})e^{i\gamma(\elementofbigGroup,i,k+1)} \pi(U^{[i,k]}_\elementofbigGroup) V^{R_{k+1}}_\elementofbigGroup \\
    &= e^{i\gamma(\elementofbigGroup,i,k+1)} \pi(U^{[i,k)}_\elementofbigGroup) \pi(U^{(k)}_{\elementofbigGroup'\elementofbigGroup(\elementofbigGroup')^{-1}} )\pi(U^{(k)}_{\elementofbigGroup'})V^{R_{k+1}}_\elementofbigGroup,
\end{align*}
since $U^{(k)}_\elementofbigGroup$ is a linear representation. The last two terms commute as above, and the proof is concluded by noting that $\elementofbigGroup'\elementofbigGroup(\elementofbigGroup')^{-1} = \elementofbigGroup$ by assumption, and~\eqref{eq: from V^R_i to V^R_j} again. 
\end{proof}

As a final step, we demonstrate the complementarity of left and right half-chain representations.
\begin{lemma}
\label{lemma: invariance under V^L V^R}
For all $i\in \mathbb{Z}$ and $\elementofbigGroup\in\bigGroup$, there is $c^i_\elementofbigGroup\in U(1)$ such that
    \begin{equation}
V_{\elementofbigGroup}^{L_i} V_{\elementofbigGroup}^{R_i} \Omega = c^i_\elementofbigGroup \Omega \label{eq: invariance under V^L V^R}
\end{equation}
\end{lemma}
Redefining $V_{\elementofbigGroup}^{L_i}\to \overline{c^i_\elementofbigGroup} V_{\elementofbigGroup}^{L_i}$ does not change the cohomology class of the projective representations, and so we may assume that $c^i_\elementofbigGroup=1$ for all $i,\elementofbigGroup$ from now on.
\begin{proof}
    The state $\omega$ is invariant under the symmetry $\beta_\elementofbigGroup$
    \begin{equation*}
        \omega
        =\omega\circ\beta_{\elementofbigGroup^{-1}}
        =\omega\circ\beta^{L_i}_{\elementofbigGroup^{-1}} \otimes\beta^{R_i}_{\elementofbigGroup^{-1}}. 
    \end{equation*}
    In the GNS representation, that translates to
    \begin{equation*}
        \langle{\Omega},{\pi(A)\Omega}\rangle=\langle {\Omega},{ \text{Ad}[V^{L_i}_{\elementofbigGroup^{-1}} V^{R_i}_{\elementofbigGroup^{-1}}]( \pi(A))\Omega}\rangle,
    \end{equation*}
    for all $A\in \mathcal{A}$. In other words, the vectors $\Omega$ and $\Omega' = V^{L_i}_\elementofbigGroup V^{R_i}_\elementofbigGroup\Omega$ are representatives of the same state $\omega$ in the representation $(\mathcal{H},\pi)$. By uniqueness of the GNS representation, there is a unitary $U$ such that $\Omega' = U\Omega$ and $U^*\pi(A) U = \pi(A)$ for all $A\in \mathcal{A}$. Here again, the claim now follows from the irreducibility of the representation.
\end{proof}

\subsection{String order for a non-degenerate SPT state}\label{sec: string order}
With this preparation, we can now prove the existence of string order. 
\begin{proof}[Proof of Proposition~\ref{prop: string order}] To match the notation of Proposition~\ref{prop: string order}, we define $(W^i_{\tilde{g}}, W^j_{\tilde{g}}) = (W^{L_i}_{\tilde{g}}, W^{R_j*}_{\tilde{g}^{-1}})$ whenever $i<j$. By Proposition~\ref{prop: existence of VRg + properties}, they satisfy
\begin{equation*}
    \pi(W^{i*}_{\tilde{g}})\Omega=V^{L_i*}_{\tilde{g}}\Omega,\qquad \pi(W^j_{\tilde{g}})\Omega=V^{R_j*}_{\tilde{g}^{-1}}\Omega
    =\overline{\mu_\omega(\tilde{g},\tilde{g}^{-1})}V^{R_j}_{\tilde{g}}\Omega.
\end{equation*}
Then, by Lemma~\ref{lemma: invariance under V^L V^R}, 
\begin{align*}
    \mu_\omega(\tilde{g},\tilde{g}^{-1}) \omega(W^{j}_{\tilde{g}})
    &
    =\langle{\Omega},{V^{R_j}_{\tilde{g}}\Omega}\rangle\\
    &=\langle{\Omega},{V^{L_j*}_{\tilde{h}}V^{R_j*}_{\tilde{h}}V^{R_j}_{\tilde{g}}V^{R_j}_{\tilde{h}}V^{L_j}_{\tilde{h}}\Omega}\rangle
\end{align*}
Now, $V^{L_j}_{\tilde{h}}$ commutes through by~\eqref{eq: V^L and V^R commutation} and cancels with its adjoint. Since $\omega$ is non-degenerate for the commuting pair $\tilde g,\tilde h$, then
\begin{equation*}
    V^{R_j}_{\tilde g} V^{R_j}_{\tilde h} = \sigma_\omega(\tilde g,\tilde h)\, V^{R_j}_{\tilde h}V^{R_j}_{\tilde g},
\end{equation*}
and so
\begin{equation*}
    \langle{\Omega},{V^{L_j*}_{\tilde{h}}V^{R_j*}_{\tilde{h}}V^{R_j}_{\tilde{g}}V^{R_j}_{\tilde{h}}V^{L_j}_{\tilde{h}}\Omega}\rangle
    = \sigma_\omega(\tilde{g},\tilde{h})\langle{\Omega},{V^{R_j}_{\tilde{g}}\Omega}\rangle
    =\sigma_\omega(\tilde{g},\tilde{h})\mu_\omega(\tilde{g},\tilde{g}^{-1})\omega(W^j_{\tilde{g}}).
\end{equation*}
Since $\sigma_\omega(\tilde{g},\tilde{h})\neq 1$, see~\eqref{eq: sigma mixed}, we conclude that $\omega(W^j_{\tilde{g}})=0$.

What is more, using Proposition~\ref{prop: existence of VRg + properties} and Lemma~\ref{lem:translating Vs} yield
\begin{align*}
        \mu_\omega(\tilde{g},\tilde{g}^{-1})\omega(W^{i}_{\tilde{g}} U^{[i,j)}_{\tilde{g}} W^{j}_{\tilde{g}})
        &=\langle{\Omega},{V^{L_i}_{\tilde{g}}\pi(U^{[i,j)}_{\tilde{g}}) V^{R_j}_{\tilde{g}} \Omega}\rangle\\
        &=e^{-i\gamma({\tilde{g}},i,j)}\langle{\Omega},{V^{L_i}_{\tilde{g}}  V^{R_i}_{\tilde{g}} \Omega}\rangle
        =e^{-i\gamma({\tilde{g}},i,j)}c^i_{\tilde{g}},
    \end{align*}
where the last step is by Lemma~\ref{lemma: invariance under V^L V^R}. This concludes the proof.
\end{proof}

\noindent It is worth pointing out here that the non-degeneracy of $\omega$ is thus responsible only for the vanishing of the endpoint expectation values.

\section{Correlations in the measured state}
\label{sec: correlations in the measured state}

In this section, we always have that $G$ is a closed Abelian subgroup of $\bigGroup$. 

\subsection{Local measurements}

Consider the projectors that define the local measurements as in Definition~\ref{def: local measurments}. As discussed there, the projectors $P_q$ are spectral projectors associated with the unitaries $U_g$. In particular, we have the following.

\begin{lemma}\label{lemma: U P on site}
Let $g \in G$ and let $P^{(j)}_q$ be as in~\eqref{eq:Projectors}. Then ${U}_g^{(j)} P^{(j)}_q = \overline{\chi_{q}(g)} P^{(j)}_q$. 
\end{lemma}
\begin{proof}
Dropping the index $j$ for simplicity, we compute
\begin{align*}
{U}_g P_q 
&= \int_G \chi_{q}(g') U_g U_{g'} \mathrm{d}\mu(g')
= \int_G \chi_{q}(g^{-1} g'') U_{g''} \mathrm{d}\mu(g^{-1} g'') \\
&= {\chi_{q}(g^{-1})} \int_G \chi_{q}(g'') U_{g''} \mathrm{d}\mu(g'')
= \overline{\chi_{q}(g)} \cdot P_q. 
\end{align*}
where we used that ${U}_g$ and $\chi_{q}(g)$ are a linear representation of $G$, and the left invariance of the Haar measure in the second line.
\end{proof}

Since the measured state $\omega_{n,\mathbf{q}}$ defined in~\eqref{eq: def of w_n} is only a local perturbation of the initial $\omega$, it is represented in the GNS space $\mathcal{H}$ of $\omega$ by the vector
\begin{equation}
    \Omega_{n,\mathbf{q}}=(\mathcal{N}_{n,\mathbf{q}})^{-\frac{1}{2}}\pi(P_{n,\mathbf{q}})\Omega.
    \label{eq: definition Omega_n}
\end{equation} 
It immediately follows from Lemma~\ref{lemma: U P on site} that $\Omega_{n,\mathbf{q}}$ is an eigenvector of the subgroup $G$ symmetry on each measured site (as it should, after measurement):
\begin{equation*}
    \pi(U^{(k)}_g) \Omega_{n,\mathbf{q}} = \overline{\chi_{q_k}(g)} \Omega_{n,\mathbf{q}}, 
\end{equation*}
for all $g\in G$ and $k\in[-n,n]$.

\subsection{Correlations in the post-measurement state}
\label{subsec: computation of correlation after measurements}
The proof of Theorem~\ref{theorem: entanglement growth} relies on a lower bound on the correlations in the state $\omega_{n,\mathbf{q}}$. In order to simplify notation, we now drop the index $\mathbf{q}$ labeling the measurement outcomes. 

\begin{lemma}\label{lemma: post meas SRE}
    The state $\omega_n$ is unitarily equivalent to $\omega$. Moreover, it is a $\tilde \bigGroup$-invariant SRE state, where $\tilde \bigGroup = Z_\bigGroup(G)$.
\end{lemma}
\noindent Note that $G\subset Z_\bigGroup(G)$ since $G$ is Abelian. Hence $\tilde \bigGroup$ is a closed subgroup containing $G$. In particular $\omega_n$ is invariant under both $\beta_{\tilde g}$ and $\beta_{\tilde h}$. Let us also emphasize that $\omega_n$ is in general not $\bigGroup$-invariant. The invariance just established is all that will be used below.
\begin{proof}
    Unitary equivalence follows abstractly from the fact that
    
   \begin{equation*}
       \omega_n(A) =\omega(P_n)^{-1}\omega(P_n AP_n)
   \end{equation*}
   is a vector state in the GNS representation of $\omega$, with vector representative given by~\eqref{eq: definition Omega_n}. $G$-invariance follows from the invariance of $\omega$ and from
    \begin{equation*}
        \beta_g(P_n) = \mathrm{Ad}[U^{[-n,n]}_g](P_n) = P_n,
    \end{equation*}
by Lemma~\ref{lemma: U P on site}. For the invariance under $Z_\bigGroup(G)$, it suffices to note that $U^{[-n,n]}_{h}$ commutes with $U^{(j)}_g$ for every $j\in\mathbb{Z}$ and every $g\in G$ by definition of the centralizer, and the fact that $U_g^{(k)}$ are representations, which similarly yields $\beta_{ h}(P_n)=P_n$ for all $h\in Z_\bigGroup(G)$. 

It remains to show that $\omega_n$ is SRE. For any $A\in\mathcal{A}_{[n+r,\infty)}\cup \mathcal{A}_{(-\infty,-n-r]}$, 
\begin{equation*}
    \vert \omega_n(A) - \omega(A)\vert 
    = \frac{1}{\omega(P_n)}\vert \omega(P_n A) - \omega(P_n)\omega(A)\vert
    \leq \frac{f(r)}{\omega(P_n)}\Vert A\Vert 
\end{equation*}
by decay of correlations, Lemma~\ref{lemma: decay of correlation}, where $f\in\mathcal{F}$ depends only on $\alpha$. Repeating the arguments of the proof of Proposition~\ref{prop: existence of W^R_g},
we obtain an almost local unitary $T_n$ such that
\begin{equation*}
    \omega_n = \omega\circ\mathrm{Ad}[T_n] = \omega_0\circ\alpha_1\circ\mathrm{Ad}[T_n].
\end{equation*}
Let $\{T_{n,s}\}_{s\in[0,1]}$ be the family interpolating between $\mathbb{I}$ and $T_n$
provided by Lemma~\ref{lemma: interpolation}, and set
\begin{equation*}
    \alpha_{n,s} := \alpha_s\circ\mathrm{Ad}[T_{n,s}],\qquad s\in[0,1].
\end{equation*}
Then $\alpha_{n,0}=\alpha_0\circ\mathrm{Ad}[\mathbb{I}]=\mathrm{id}$ and
$\omega_0\circ\alpha_{n,1}=\omega_n$, while
$s\mapsto\alpha_{n,s}$ is strongly continuous because $s\mapsto\alpha_s$ is and
$s\mapsto T_{n,s}$ is norm-continuous. Finally, each $\mathrm{Ad}[T_{n,s}]$ satisfies the
factorization property for a decay function depending only on $\hat f_n$, see
Lemma~\ref{lemma: inner implies split} and Lemma~\ref{lemma: interpolation}, and each $\alpha_s$
does so for $f$. Hence, for any $n\in\mathbb{N}$, Lemma~\ref{lemma: composition split} implies that there is $\bar f_n\in\mathcal{F}$ such that every $\{\alpha_{n,s}:s\in[0,1]\}$ satisfies the factorization property for $\bar f_n$. Therefore $\omega_n$ is
short-range entangled in the sense of Definition~\ref{def: short range entangled states}.
\end{proof}

\begin{lemma}\label{lemma: inner implies split}
    Let $T\in\mathcal{U}$ be $f$-localized. Then the automorphism $\mathrm{Ad}[T]$
    satisfies the factorization property, Definition~\ref{def:split automorphism},
    with a decay function depending only on $f$.
\end{lemma}

\begin{proof}
W.l.o.g. we assume that $T$ is $f$-close to $\{0\}$ for some $f\in\mathcal{F}$, namely
\begin{equation*}
    \|T-\Pi_{[-r,r]}(T)\|\leq f(r), \quad \forall r>0.
\end{equation*}
We first show that $T$ can be approximated by strictly local unitaries. Indeed, the polar decomposition of $\Pi_{[-r,r]}(T)\in \mathcal{A}_{[-r,r]}$ reads 
\begin{equation*}
    \Pi_{[-r,r]}(T)=T_rM_r,
\end{equation*}
with $T_r$ a unitary and $M_r$ a positive semi-definite self-adjoint operator. Let $\{m_{r,k}\}_{k}$ be the eigenvalues of $M_r$. Note that
\begin{equation*}
    \|M_r^2-\mathbb{I}_{\mathcal{A}_{[-r,r]}}\|=\|\Pi_{[-r,r]}(T)\Pi_{[-r,r]}(T)^*-TT^*\|\leq 2f(r),
\end{equation*}
therefore, $\max_k|m_{r,k}^2-1|\leq 2f(r)$. Let $r_0$ be such that $2f(r_0)<1$. Then the same inequality holds for all $r\geq r_0$, which implies that $m_{r,k}>0$ for all $k$. Furthermore, 
\begin{equation*}
    \|\Pi_{[-r,r]}(T) -T_r\|=\|M_r-\mathbb{I}_{\mathcal{A}_{[-r,r]}}\|=\max_k|m_{r,k}-1|\leq 2f(r),
\end{equation*}
and so
\begin{equation*}
    \|T_r-T\|\leq 3 f(r). 
\end{equation*}

Now, the factorization of the automorphism $\mathrm{Ad}[T]=\alpha^{L_j}\otimes\alpha^{R_j}\circ \mathrm{Ad}[U^j]$ will depend on the location of the cut $j\in \mathbb{Z}$. Let $r_0$ be as above.

\noindent Case (i): $j\in [-r_0, r_0]$. We pick $U^j\coloneqq T$ and $\alpha^{L_j}=\alpha^{R_j}\coloneqq\mathrm{id}$. If $r>r_0$, then $[-r+|j|,r-|j|]\subset [j-r, j+r]$, leading to 
    \begin{align*}
        \|T-\Pi_{[j-r,j+r]}(T)\|%&= \|T - \Pi_{[-r+|j|,r-|j|]}(T) + \Pi_{[-r+|j|,r-|j|]}(T) -\Pi_{[j-r,j+r]}(T)\|\\
        &\leq \|T - \Pi_{[-r+|j|,r-|j|]}(T)\|\\ &\quad+\|\Pi_{[j-r,j+r]}( \Pi_{[-r+|j|,r-|j|]}(T) -T)\| \\
        %&\leq 2 \|T - \Pi_{[-r+|j|,r-|j|]}(T)\|\\
        &\leq 2f(r-|j|)\leq 2f(r-r_0). 
    \end{align*}

\noindent Case (ii): $j>r_0$. Since $T_{j-1}\in \mathcal{A}_{[-j+1, j-1]}$, define
    \begin{equation*}
        \alpha^{L_j}\coloneqq\mathrm{Ad}[T_{j-1}], \quad\alpha^{R_j}\coloneqq\mathrm{id}, \quad  U^j\coloneqq TT_{j-1}^* .%\prod_{r=j}^{\infty}T_r
    \end{equation*}
    Note that
    \begin{equation*}
        \|U^j-\mathbb{I}\|=\|T-T_{j-1}\|\leq 3f(j-1).
    \end{equation*}
    and so,
    \begin{equation*}
        \|U^j-\Pi_{[j-r,j+r]}(U^j)\|\leq 6f(j-1)\leq 6f(r/2), \quad \forall r\leq 2j. 
    \end{equation*}
    Let $r>2j$, then $[j-r, -j+ r]\subset [j-r, j+r]$: 
    \begin{align*}
        \|U^j-T_{r-j}T^*_{j-1}\|\leq 3f(r-j)
    \end{align*}
    and $T_{r-j}T^*_{j-1}\in\mathcal{A}_{[j-r, j+r]}$, therefore, 
    \begin{equation*}
        \|\Pi_{[j-r, j+r]}(U^j)-T_{r-j}T^*_{j-1}\|\leq 3f(r-j).
    \end{equation*}
    We conclude that
    \begin{equation*}
        \|U^j- \Pi_{[j-r, j+r]}(U^j)\| \leq 6f(r-j)\leq 6f(r/2), \quad \forall r>2j. 
    \end{equation*}

\noindent Case (iii): $j<-r_0$. The argument is as above with the choices $\alpha^{L_j}\coloneqq\mathrm{id}$, $\alpha^{R_j}\coloneqq\mathrm{Ad}[T_{|j|}]$ and $U^j= TT_{|j|}^*$.
\end{proof}

\begin{lemma}\label{lemma: interpolation}
    Let $T\in\mathcal{U}$ be $f$-close to $\{0\}$. Then there is a norm-continuous
    family $\{T_s\}_{s\in[0,1]}\subset\mathcal{U}$ with
    \begin{equation*}
        T_0 = \mathbb{I},\qquad T_1 = T,
    \end{equation*}
    such that every $T_s$ is $\hat f$-close to $\{0\}$ for a single $\hat f\in\mathcal{F}$
    depending only on $f$.
\end{lemma}
 
\begin{proof}
Let $r_0$ be as in the proof of Lemma~\ref{lemma: inner implies split}, namely
$2f(r_0)<1$, and for $r\geq r_0$ let $T_r\in\mathcal{U}_{[-r,r]}$ be the unitary
obtained there from the polar decomposition of $\Pi_{[-r,r]}(T)$, so that
\begin{equation}\label{eq: Tr close to T}
    \|T_r - T\| \leq 3f(r).
\end{equation}
Set $r_k = r_0+k$ and $S_k = T_{r_k}\in\mathcal{U}_{[-r_k,r_k]}$ for $k\geq 0$. By
\eqref{eq: Tr close to T},
\begin{equation*}
    \|S_k^*S_{k+1}-\mathbb{I}\| = \|S_{k+1}-S_k\| \leq 6f(r_k).
\end{equation*}
Let $k_1$ be the smallest integer with $6f(r_{k_1})<2$. Just like $r_0$, $k_1$ depends only on $f$. For $k\geq k_1$ the unitary $S_k^*S_{k+1}$ lies in the finite-dimensional algebra
$\mathcal{A}_{[-r_{k+1},r_{k+1}]}$, and its spectrum is contained in
$\{\mathrm{e}^{\mathrm{i}\theta} : |\theta|\leq 3\pi f(r_k)\}$, since
$|\mathrm{e}^{\mathrm{i}\theta}-1|\leq\delta<2$ forces
$|\theta|\leq 2\arcsin(\delta/2)\leq \pi\delta/2$. Since $3\pi f(r_k)<\pi$, the principal logarithm is defined by the spectral theorem in that matrix algebra and yields
$K_k = K_k^*\in\mathcal{A}_{[-r_{k+1},r_{k+1}]}$ with
\begin{equation*}
    S_k^* S_{k+1} = \mathrm{e}^{\mathrm{i}K_k},
    \qquad \|K_k\| \leq 3\pi f(r_k).
\end{equation*}
Since $\mathcal{A}_{[-r_{k_1},r_{k_1}]}$ is a full matrix
algebra, we may likewise write $S_{k_1}=\mathrm{e}^{\mathrm{i}K}$ with
$K=K^*\in\mathcal{A}_{[-r_{k_1},r_{k_1}]}$.
 
We now choose an increasing sequence $\tfrac12 = \tau_{k_1+1}<\tau_{k_1+2}<\cdots$
converging to $1$, and define
\begin{equation*}
    T_s :=
    \begin{cases}
        \mathrm{e}^{2\mathrm{i}sK}, & s\in[0,\tfrac12],\\[4pt]
        S_k\,\mathrm{e}^{\mathrm{i}t_k(s)K_k}, & s\in[\tau_{k+1},\tau_{k+2}],\ k\geq k_1,\\[4pt]
        T, & s=1,
    \end{cases}
\end{equation*}
where $t_k$ maps $[\tau_{k+1},\tau_{k+2}]$ affinely onto $[0,1]$. Then $T_0=\mathbb{I}$,
$T_{1/2}=S_{k_1}$, and $T_{\tau_{k+2}} = S_k\mathrm{e}^{\mathrm{i}K_k}=S_{k+1}$, so the
family is well defined and norm-continuous on $[0,1)$. It is continuous at $s=1$ because,
for $s$ in the $k$-th leg,
\begin{equation}\label{eq: leg close to T}
    \|T_s - T\| \leq \|T_s - S_k\| + \|S_k - T\|
    \leq \|K_k\| + 3f(r_k) \leq 3(\pi+1)f(r_k)\xrightarrow[k\to\infty]{}0 .
\end{equation}
 
It remains to exhibit a uniform decay function. For all $s\in[0,1)$, the operator $T_s$ is strictly local, with $T_s\in\mathcal{A}_{[-r_{k+1},r_{k+1}]}$ on the $k$-th leg, so that $\|T_s-\Pi_{[-r,r]}(T_s)\|=0$ as soon as $r\geq r_{k+1}$. Suppose then that $r>r_{k_1}$ and that $s$ lies in the $k$-th leg with
$r<r_{k+1}$, so that $r_k\geq r-1$. Using
$\|A-\Pi_\Lambda(A)\|\leq 2\|A-B\|+\|B-\Pi_\Lambda(B)\|$ with $B=T$, together with
\eqref{eq: leg close to T} and $\|T-\Pi_{[-r,r]}(T)\|\leq f(r)$,
\begin{equation*}
    \|T_s - \Pi_{[-r,r]}(T_s)\| \leq 6(\pi+1)f(r-1) + f(r) \leq (6\pi+7)f(r-1).
\end{equation*}
Hence every $T_s$, including $T_1=T$, is $\hat f$-close to $\{0\}$, where
\begin{equation*}
    \hat f(r) :=
    \begin{cases}
        2, & r\leq r_{k_1},\\
        \min\bigl\{2,\ (6\pi+7)f(r-1)\bigr\}, & r> r_{k_1},
    \end{cases}
\end{equation*}
depends only on $f$, and $\hat f\in\mathcal{F}$
\end{proof}

\begin{lemma}\label{lemma: composition split}
    Let $\alpha,\beta$ be $\ast$-automorphisms of $\mathcal{A}$ satisfying the
    factorization property with decay functions $f_\alpha$ and $f_\beta$. Then
    $\alpha\circ\beta$ satisfies the factorization property, with a decay function
    depending only on $f_\alpha$ and $f_\beta$.
\end{lemma}
 
\begin{proof}
Fix $j\in\mathbb{Z}$ and write
$\alpha = (\alpha^{L_j}\otimes\alpha^{R_j})\circ\mathrm{Ad}[U^j]$ and
$\beta = (\beta^{L_j}\otimes\beta^{R_j})\circ\mathrm{Ad}[V^j]$ as in
Definition~\ref{def:split automorphism}. Setting
$\gamma := \beta^{L_j}\otimes\beta^{R_j}$ and using
$\mathrm{Ad}[U^j]\circ\gamma = \gamma\circ\mathrm{Ad}[\gamma^{-1}(U^j)]$,
\begin{equation*}
    \alpha\circ\beta
    = \bigl( (\alpha^{L_j}\circ\beta^{L_j}) \otimes (\alpha^{R_j}\circ\beta^{R_j}) \bigr)
      \circ \mathrm{Ad}\bigl[\gamma^{-1}(U^j)\,V^j\bigr].
\end{equation*}
The first factor preserves $\mathcal{A}_{L_j}$ and $\mathcal{A}_{R_j}$. Moreover
$\gamma = \beta\circ\mathrm{Ad}[V^{j*}]$ is almost local by Lemma~\ref{lemma: split implies quasi local}
and the remark following it, hence so is $\gamma^{-1}$; consequently $\gamma^{-1}(U^j)$ is
$f'$-close to $\{j\}$ with $f'$ depending only on $f_\alpha,f_\beta$, and so is
$\gamma^{-1}(U^j)V^j$. As all bounds are uniform in $j$, the claim follows.
\end{proof}

\begin{remark}
    While all post-measurement states $\omega_n$ are indeed SRE, the automorphisms $\alpha_n$ constructed above are less and less local as $n$ increases. What we will show below, which is really the claim of Theorem~\ref{theorem: entanglement growth}, is that it must be so: There is no sequence of factorizable automorphisms $\alpha_n$ that are uniformly local in the sense that the unitaries in~\eqref{eq: split-automorphism} are $f$-local where $f$ can be chosen uniformly in $n$.
\end{remark}

With this lemma, Proposition~\ref{prop: existence of W^R_g} yields unitaries $W^{R_j}_{{g},n}$ for all $g\in G$ such that 
\begin{equation*}
    \omega_n\circ\beta^{R_j}_{{g}}= \omega_n\circ\mathrm{Ad}[W^{R_j}_{{g},n}],
\end{equation*}
where $W^{R_j}_{{g},n}$ is almost local around $\{j\}$. That is, there is some function $f_{W_n}\in \mathcal{F}$ that depends only on $f_{\alpha_n}$ such that for all $r\geq0$,
\begin{equation*}
    \|W^{R_j}_{{g},n} -\Pi_{[j-r,j+r]}(W^{R_j}_{{g},n})\|\leq f_{W_n}(r).
\end{equation*}

These new operators have similar properties as $W^{R_j}_{{g}}$ for the state $\omega$, most importantly
\begin{equation}\label{eq:VWn}
    V^{R_j*}_{{g}} \Omega_n =\pi(W^{R_j*}_{{g},n})\Omega_n ,
\end{equation}
with $V^{R_j}_{{g}}$ the unitary defined in Proposition~\ref{prop: existence of VRg + properties}. Indeed, \eqref{eq: V^R_g properties} implies that
\begin{equation*}
    \langle\Omega_n, \mathrm{Ad}[V^{R_j}_g]( \pi(A) )\Omega_n\rangle
    = \omega_n\circ\beta^R_g(A)
    =\langle\Omega_n, \mathrm{Ad}[\pi(W^{R_j}_{{g},n})](\pi(A))\Omega_n\rangle
\end{equation*}
Once again, this implies~\eqref{eq:VWn} up to a phase, which can be absorbed in the definition of $W^{R_j}_{{g},n}$.

\begin{remark}
As pointed out earlier, $\omega_n$ is a vector state in the same GNS representation as $\omega$. In particular, we may choose to pick $V^{R_j*}_{{g}}$ in~\eqref{eq:VWn} to be those obtained from $\omega$, and the intertwining relation~\eqref{eq: V^R_g properties} is a property of $\pi$ alone. In particular $\sigma_{\omega_n}=\sigma_\omega$: the measurement does not change the index. Incidentally, this shows the need to take the limit $n\to\infty$ to make a precise mathematical statement about the entanglement transition.
\end{remark}

Now, Lemma~\ref{lemma: invariance under V^L V^R} relies only on the invariance of $\omega$ under the group $\bigGroup$. It follows that the lemma continues to hold for the vector $\Omega_n$ with $n$-dependent phases, but only for the post-measurement symmetry group $\tilde\bigGroup$, see Lemma~\ref{lemma: post meas SRE}. Since this, Proposition~\ref{prop: existence of VRg + properties} and Lemma~\ref{lem:translating Vs} are all that enter the proof of Proposition~\ref{prop: string order} in Section~\ref{sec: string order}, that result continues to hold for $\omega_n$. The only  additional fact is that the measurement of $G$ yields a joint eigenstate of all $\{U^{(j)}_g:g\in G\}$ in the measurement region, namely the string reduces to a scalar.

Specifically, let $\tilde{g},\tilde{h}$ be such that $\sigma_\omega(\tilde{g}, \tilde{h})\neq 1$. Similarly to Proposition~\ref{prop: string order}, we define $(W^i_{n}, W^j_{n})\coloneqq(W^{L_i}_{\tilde{g},n},W^{R_j*}_{\tilde{g}^{-1},n})$ whenever $i<j$. Then
\begin{equation*}
    \omega_n(W^{j}_{n})=0,\qquad 
    \vert \omega_n(W^i_{n} U^{[i,j)}_{\tilde{g}}W^{j}_n)\vert = 1,
\end{equation*}
for all $n\in\mathbb{N}$. Furthermore, if $-n\leq i<j\leq n$, then by~\eqref{eq: V and U commutation},
\begin{align*}
    \omega_n(W^i_{n} U^{[i,j)}_{\tilde{g}}W^{j}_n)
    &=\overline{\mu_\omega(\tilde{g},\tilde{g}^{-1})}\langle\Omega_n, V^{L_i}_{\tilde{g}}  V^{R_j}_{\tilde{g}}\pi (U^{[i,j)}_{\tilde{g}}) \Omega_n\rangle \\
    &= \overline{\mu_\omega(\tilde{g},\tilde{g}^{-1})}\prod_{k=i}^{j-1}\overline{\chi_{q_k}}(\tilde g) \langle\Omega_n, V^{L_i}_{\tilde{g}}V^{R_j}_{\tilde{g}}  \Omega_n\rangle \\
    &= \prod_{k=i}^{j-1}\overline{\chi_{q_k}}(\tilde g)\omega_n(W^i_{n} W^{j}_n)
\end{align*}
where we used Lemma~\ref{lemma: U P on site} in the second equality. In particular $\vert \omega_n(W^i_{n} W^{j}_n)\vert = 1$ and thus we have proved the following:
\begin{lemma}\label{lemma:LRO n}
    For any $-n\leq i<j\leq n$
    \begin{equation}\label{eq: length n correlation for measured state}
    |\omega_n(W^i_nW^j_n)- \omega_n(W^i_n)\omega_n(W^j_n)|  =1.
\end{equation}
\end{lemma}

\subsection{Short-range entanglement bound}
\label{subsec: SRE bound after measurements}

Ideally, one would wish to extend~\eqref{eq: length n correlation for measured state} to an infinite volume state $\nu$ and fixed observables $ O^{i}_{\tilde{g}}, O^{j}_{\tilde{g}}$ that are almost localized near $\{i\}$ respectively $\{j\}$. While we can do that in the case of QCAs, see Section~\ref{sec: long range correlation for QCA} below, our locality estimates for the unitaries $(W^i_n,W^j_n)$ deteriorate too badly for an SRE state obtained from a general factorizable automorphism $\alpha$. We therefore study the locality of the sequence $(\omega_n)_{n\in\mathbb{N}}$ and the operators $(W^i_n)_{n\in\mathbb{N}},(W^j_n)_{n\in\mathbb{N}}$ in more detail.

\begin{lemma}\label{lemma:Decay of corr n}
    Let $(W^i_n,W^j_n)$ be as above with $-n\leq i<j\leq n$ and $|j-i|=n$. Then
\begin{equation}\label{eq: decaying correlation for measured state}
    |\omega_n(W^i_nW^j_n)- \omega_n(W^i_n)\omega_n(W^j_n)| \leq 4f_{\alpha_n}(\frac{n}{3})+4f_{W_n}(\frac{n}{3}). 
\end{equation}
\end{lemma}

\begin{proof}
    First let us use the fact that $W^i_n$ and $W^j_n$ are almost local around $i$ and $j$ respectively. Let $I_n=[i-\frac{n}{3}, i+\frac{n}{3}]$ and $J_n=[j-\frac{n}{3},j+\frac{n}{3}]$. Following the proof of Proposition~\ref{prop: existence of W^R_g}, we have that
    \begin{equation}\label{eq: ineq1}
       \Vert  W^i_n - \Pi_{I_n}(W^i_n)\Vert\leq f_{W_n}(\frac{n}{3})
    \end{equation}
 and in turn that
    \begin{equation}\label{eq: ineq2}
        \Vert W^i_nW^j_n - \Pi_{I_n}(W^i_n)\Pi_{J_n}(W^j_n)\Vert
        \leq 2f_{W_n}(\frac{n}{3}).
    \end{equation}
    Note that $f_{W_n}\in\mathcal{F}$ depends only on $f_{\alpha_n}$.

Moreover, the state $\omega_n$ is also short-range entangled and the distance between the support of $\Pi_{I_n}(W^i_n)$ and $\Pi_{J_n}(W^j_n)$ is at least $n/3$. Therefore, decay of correlations, Lemma~\ref{lemma: decay of correlation}, yields
\begin{equation}\label{eq: ineq3}
    |\omega_n(\Pi_{I_n}(W^i_n)\Pi_{J_n}(W^j_n)) -\omega_n(\Pi_{I_n}(W^i_n))\omega_n(\Pi_{J_n}(W^j_n))| \leq 4f_{\alpha_n}(\frac{n}{3}).
\end{equation}
We conclude the proof by combining the three inequalities~(\ref{eq: ineq1},\ref{eq: ineq2},\ref{eq: ineq3}).
\end{proof}

While we proved in Lemma~\ref{lemma: post meas SRE} that $\omega_n$ is short-range entangled and noted that our locality upper bound on $\alpha_{n,1}$ deteriorates with $n$, that proof per se does not show that it must be so. In particular, the $\alpha_{n}$ used in the proof above needs not be the one exhibited before: It is any approximately factorizable automorphism such that $\omega_n = \omega_0 \circ\alpha_n$. It is only in combination with Lemma~\ref{lemma:LRO n} that a meaningful statement arises.

\subsection{Proof of Theorem~\ref{theorem: entanglement growth}}\label{sec:Proof of Main 1}

The combination of Lemma~\ref{lemma:LRO n} and Lemma~\ref{lemma:Decay of corr n} yields the bound
\begin{equation}\label{Central inequality}
    1\leq 4f_{\alpha_n}(\frac{n}{3}) + 4f_{W_n}(\frac{n}{3})
\end{equation}
for the fast decaying functions $f_{\alpha_n},f_{W_n}\in\mathcal{F}$, and for all $n\in\mathbb{N}$.

Let us now assume by contradiction that there is $f\in\mathcal{F}$ such that $f_{\alpha_n}\leq f$. Since the map $f_\alpha \mapsto f_W$ is increasing on $\mathcal{F}$, see~\eqref{eq: F of f almost locality of WR}, we conclude that in turn $f_{W_n}\leq g$ for some $g\in\mathcal{F}$, uniformly in $n$. We conclude from this and~\eqref{Central inequality} that there is $F\in\mathcal{F}$ such that
\begin{equation*}
    1\leq F(n)
\end{equation*}
for all $n\in\mathbb{N}$, which is a contradiction. \hfill $\Box$

\section{Strict locality and long-range correlations}
\label{sec: long range correlation for QCA}

A stronger result holds if the unitary factorizing an almost local automorphism~\eqref{eq: split-automorphism} is strictly local uniformly for all sites $j$. It follows that the spread of the support of local observables is strictly finite, equivalently $f_\alpha(r) = 0$ in Lemma~\ref{lemma: split implies quasi local} for $r$ sufficiently large.

In this section, we will show that this implies that the operators $W^{R_j}_\elementofbigGroup$ constructed in Proposition \ref{prop: existence of W^R_g} are also strictly local. We prove Theorem~\ref{theorem: long range correlations with QCA} by showing that arbitrary long-range correlations can be obtained by applying measurements to non-overlapping blocks of spins. 

\subsection{Strict locality}

Proposition~\ref{prop: existence of W^R_g} is strengthened when assuming that the automorphism $\alpha$ has a finite spread, i.e. it is a Quantum Cellular Automaton (QCA), see Definition~\ref{def: QCA}. In that case, the operators $W^{L_i}_\elementofbigGroup$ and $W^{R_i}_\elementofbigGroup$ are strictly local.

\begin{lemma}\label{lma:Local Ws}
   Let $\omega$ be a $\bigGroup$-invariant SRE state and such that $\omega = \omega_0\circ\alpha$, where $\alpha$ is a QCA. Then there is $N\in\mathbb{N}$ such that for each $j\in\mathbb{Z}$ and each $\elementofbigGroup \in \mathcal{G}$, there exists a unitary $W^{R_j}_\elementofbigGroup \in \mathcal{A}_{[j-N,j+N]}$ such that
\begin{equation*}
\omega \circ \beta^{R_j}_\elementofbigGroup  = \omega \circ \mathrm{Ad}[W^{R_j}_\elementofbigGroup].
\end{equation*}
The same holds for $\beta^{L_i}_\elementofbigGroup$ and a unitary $W^{L_i}_\elementofbigGroup \in \mathcal{A}_{[i-1-N, i-1+N]}$.
   \label{lemma: existence and exact locality of W^R for QCA}
\end{lemma}

\begin{proof}
    Proposition~\ref{prop: existence of W^R_g} states that there is some $W^{R_j}_\elementofbigGroup$ that is $f_W$-localized for some function $f_W\in \mathcal{F}$. Moreover,  \eqref{eq: F of f almost locality of WR} gives  an analytical expression for $f_W(r)=F(f_\alpha)(r)$. Because $\alpha$ is a QCA, there is some $N\in\mathbb{N}$ such that for all $r> \frac{N}{2}$, $f_\alpha(r)=0$. It follows that for all $r > N$, we have $F(f_\alpha)(r)=0$, proving that $W^{R_j}_\elementofbigGroup$ is strictly local. 
\end{proof}

\subsection{Block measurements}
\label{subsec: Gauging the G charge on blocks}
We perform $G$-charge measurements on a block of spins $[k-N, k+N]$, which corresponds to the support of $W^{R_k}_{\tilde{g}}$, see Lemma~\ref{lemma: existence and exact locality of W^R for QCA}. The action of the symmetry on a block is given by
\begin{equation*}
        \tilde{U}^{(k)}_g = \bigotimes_{l=k-N}^{k+N} U^{(l)}_g,
\end{equation*}
and the measurement corresponds to the projector 
\begin{equation}\label{eq: definition of P on blocks}
    \tilde P_q^{(k)} = \int_G\chi_{q}(g)  \tilde{U}^{(k)}_gd\mu(g). 
\end{equation}
Let us define the set of sites of the form
\begin{equation*}
    K_n = \{(2N+1)l,~ -n\leq l \leq n\}
\end{equation*}
For $k\neq k'\in K_n$, $\tilde P_q^{(k)}$ and $\tilde P_q^{(k')}$ have disjoint support. Moreover, the support of $\tilde P_{n,\bf q} = \bigotimes_{k\in K_n}\tilde P_{q_k}^{(k)}$ spans all the sites between $-(2N+1)n -N$ and $(2N+1)n+N$. Every site in this interval undergoes exactly one measurement, but they are measured in finite blocks of length $2N+1$.

As in Lemma~\ref{lemma: U P on site},
\begin{equation}\label{eq: U P block}
    \tilde{U}^{(k)}_g\tilde P_q^{(k)}= \overline{\chi}_{q}(g)\tilde P_q^{(k)}, \quad \forall g\in G, ~ k\in K_n. 
\end{equation}
The post-measurement state is now given by 
\begin{equation}\label{eq: def of w_n on blocks}
    \tilde \omega_{n,\bf q}(A)=\frac{1}{\tilde{\mathcal{N}}_{n,\bf q}}\omega\left(\tilde P_{n,\bf q} A \tilde P_{n,\bf q} \right)
\end{equation}
with $\tilde{\mathcal{N}}_{n,\bf q}=\omega(\tilde P_{n,\bf q})$, see~\eqref{eq: def of w_n}. Its GNS representation $(\mathcal{H}, \pi, \tilde \Omega_{n,\bf q})$ is given as in~\eqref{eq: definition Omega_n}, but with $\tilde\cdot$.

\subsection{Action of the local operator on the measured state}
The fact that measurements have been blocked according to the QCA's range allows for one significant simplification in the construction of Section~\ref{subsec: computation of correlation after measurements}: We shall now show that the $n$-independent operators $W^{R_j}_{g}$ defined from the state $\omega$ before measurement continue to play their role for all the post-measurement state $\tilde \omega_{n,\bf q}$. As we did there, we now drop the index $\bf q$.

\begin{lemma}
\label{lemma: W^R acts the same on measured state}
Let $(\mathcal{H}, \pi ,\tilde \Omega_n)$ be the GNS representation of $\tilde \omega_n$. Let $W^{R_j*}_{\elementofbigGroup}$ be as in Lemma~\ref{lma:Local Ws}. Assume that $\sigma(\tilde g,g) = 1$, for all $g\in G$, see~\eqref{eq: M3}. Then for all $j\in K_n$ and $g,\tilde g$ in $G$, we have
\begin{equation}\label{eq: UWU Omega=W Omega}
\pi( \tilde{U}_{g}^{(j)*}W^{R_j*}_{\tilde{g}}\tilde{U}_{g}^{(j)} ) \Omega = \pi(W^{R_j*}_{\tilde{g}}) \Omega.
\end{equation}
Furthermore,
\begin{equation}
   \pi(W_{\tilde{g}}^{R_j*}) \tilde\Omega_n = V_{\tilde{g}}^{R_j*} \tilde \Omega_n,
   \label{eq: W^R Omega_n = V^R Omega_n}
\end{equation}
for any $j \in K_n$, where $V_{\tilde{g}}^{R_j}$ are the operators defined in Proposition~\ref{prop: existence of VRg + properties}.
\end{lemma}

\begin{proof}

Let $(\mathcal{H}, \pi, \Omega)$ be the GNS representation of $\omega$. We note first that~\eqref{eq: M3} and the Abelianness of $G$ imply that $[V^{R_j}_{ g}, V^{R_j}_{\tilde g}] = 0$ for all $g,\tilde g\in G$.
Since the support of $W^{R_j}_{\tilde{g}}$ lays outside of $L_{j-N}$ and $R_{j+N+1}$, by \eqref{eq: V and op commutating}, 
\begin{equation*}
    [V^{L_{j-N}}_{g}, \pi(W^{R_j}_{\tilde{g}})] = 0 \quad\mathrm{and}\quad [V^{R_{j+N+1}}_{g}, \pi(W^{R_j}_{\tilde{g}})] = 0.
\end{equation*}
Therefore, using Lemma~\ref{lem:translating Vs} and Lemma~\ref{lemma: invariance under V^L V^R},
\begin{align*}
\pi(W^{R_j*}_{\tilde{g}}) \Omega 
&= V^{R_j*}_{\tilde{g}} \Omega 
= \overline{c^j_g} V^{R_j*}_{\tilde{g}} V^{L_j*}_{g} V^{R_j*}_{g} \Omega
= \overline{c^j_g} V^{L_j*}_{g} V^{R_j*}_{g} V^{R_j*}_{\tilde{g}} \Omega \\
&= \overline{c^j_g} V^{L_j*}_{g} V^{R_j*}_{g} \pi(W^{R_j*}_{\tilde{g}}) \Omega.
\end{align*}
Then, the identity~\eqref{eq: from V^R_i to V^R_j} yields a $c\in U(1)$ depending on $g,i,j$, such that
\begin{equation*}
V_{g}^{L_j*} V_{g}^{R_j*} \pi(W^{R_j*}_{\tilde{g}}) \Omega 
= c V^{L_{j-N}*}_{g} V^{R_{j+N+1}*}_{g} \pi( U_{g}^{[j-N,j+N]*} )  \pi(W^{R_j*}_{\tilde{g}}) \Omega 
\end{equation*}
By~\eqref{eq: V and U commutation}, all factors now commute with each other, and so
\begin{equation*}
V_{g}^{L_j*} V_{g}^{R_j*} \pi(W^{R_j*}_{\tilde{g}}) \Omega 
= c  \pi( U_{g}^{[j-N,j+N]*} )  \pi(W^{R_j*}_{\tilde{g}}) V^{L_{j-N}*}_{g} V^{R_{j+N+1}*}_{g} \Omega 
\end{equation*}
It remains to use~\eqref{eq: from V^R_i to V^R_j} again to conclude that
\begin{equation*}
V_{g}^{L_j*} V_{g}^{R_j*} \pi(W^{R_j*}_{\tilde{g}}) \Omega
=c^j_g\pi( U_{g}^{[j-N,j+N]*} )  \pi(W^{R_j*}_{\tilde{g}}) \pi( U_{g}^{[j-N,j+N]} ) \Omega,
\end{equation*}
which is~\eqref{eq: UWU Omega=W Omega}.

With the definition~\eqref{eq: definition of P on blocks}, this implies that $\pi(W^{R_j*}_{\tilde{g}})\pi(\tilde P^{(j)})\Omega= \pi( \tilde P^{(j)} ) \pi(W^{R_j*}_{\tilde{g}})\Omega$. Since $P^{(k)}$ and $W^{R_j}_{\tilde{g}}$ commute for all $k\neq j \in K_n$ because they have disjoint supports, we conclude that
\begin{equation*}
\pi(W_{\tilde{g}}^{R_j*}) \pi( \tilde P_n )\Omega 
= \pi( \tilde P_n ) \pi(W_{\tilde{g}}^{R_j*})\Omega
= \pi( \tilde P_n ) V_{\tilde{g}}^{R_j*}\Omega
\end{equation*}
for all $j\in K_n$. Up to the normalization, this is~\eqref{eq: W^R Omega_n = V^R Omega_n}.
\end{proof}

\begin{remark}
    (i) One may wish to reproduce the proof above in the case where $W^{R_j}_{\tilde{g}}$ are only $f_W$-localized. Then \eqref{eq: UWU Omega=W Omega} and \eqref{eq: W^R Omega_n = V^R Omega_n} are no longer exact equalities but rather almost equalities with a bound that can be calculated. For example,
    \begin{equation*}
        \| \pi(W^{R_j*}_{\tilde{g}})\Omega - \pi(\tilde{U}^{(j)}_gW^{R_j*}_{\tilde{g}}\tilde{U}^{(j)*}_g)\Omega \| \leq 2f_W(N)
    \end{equation*}
    where $N$ is the radius of the measurement blocks. This follows by approximating $W^{R_j}_{\tilde{g}}$ with a strictly local operator supported on $[j-N,j+N]$. Similarly, \eqref{eq: W^R Omega_n = V^R Omega_n} becomes
    \begin{equation*}
        \| V^{R_j*}_{\tilde{g}}\tilde \Omega_n - \pi(W^{R_j*}_{\tilde{g}})\tilde \Omega_n \|\leq 4\frac{f_W(N)}{(\tilde{\mathcal{N}}_n)^\frac{1}{2}}
    \end{equation*}
    where $\tilde{\mathcal{N}}_n$ comes from the renormalization of the vector $\Omega_n$. As $\tilde{\mathcal{N}}_n$ decays in general exponentially in $n$, this bound is of little value since $N$ is fixed. This `small denominator' problem is the technical hurdle to prove Theorem~\ref{theorem: long range correlations with QCA} for the general short-range entanglement definition \ref{def: short range entangled states}.

\noindent (ii) The result of Theorem~\ref{theorem: long range correlations with QCA} is fundamentally related to the commutation of the measurement and the edge operators $W^{R_j}_g$. Indeed, with the additional assumption that 
    \begin{equation*}
    \omega(P_{n,\textbf{q} }\mathrm{Ad}[W^{R_j}_{{g}}](A)P_{n,\textbf{q} })= \omega\circ \mathrm{Ad}[W^{R_j}_{{g}}](P_{n,\textbf{q} }AP_{n,\textbf{q} }), \quad \forall A \in \mathcal{A},
\end{equation*}
    Lemma~\ref{lemma: W^R acts the same on measured state} holds for any factorizable automorphism $\alpha$. In that case, Lemma~\ref{lemma: finite correlation measured QCA} below and Theorem~\ref{theorem: long range correlations with QCA} apply without additional assumption. 
\end{remark}

\subsection{Arbitrary long-range correlations}\label{sec:QCA post correlations}

In the previous section, \eqref{eq: W^R Omega_n = V^R Omega_n} shows that the strictly local operator originally defined with respect to $\omega$ has the same action on $\tilde \omega_n$, namely
\begin{align*}
    \tilde \omega_n \circ\beta^{R_j}_{{g}}(A) 
    &=\langle\tilde\Omega_n,V_{g}^{R_j} \pi(A) V_{g}^{R_j*}\tilde\Omega_n\rangle \\
    &= \langle\tilde\Omega_n,\pi(W_{g}^{R_j}) \pi(A) \pi(W_{g}^{R_j*})\tilde\Omega_n\rangle
    = \tilde \omega_n\circ\mathrm{Ad}[W^{R_j}_{{g}}](A).
\end{align*}
We continue to follow the reasoning of Section~\ref{subsec: computation of correlation after measurements} and prove that $W^{R_j}_{\tilde{g}}$ are correlated over arbitrary long distance after the block measurements. For any $i<j\in K_n$, let the pair $(\widetilde{W}^{i}_{\tilde{g}},\widetilde{W}^{j}_{\tilde{g}})$ be as
\begin{equation}\label{eq: W^R tilde definition} 
\begin{split}
    &\widetilde{W}^{i}_{\tilde{g}}= {W}^{L_i}_{\tilde{g}}U^{[i-N-1,i)*}_{\tilde{g}}, \\
    &\widetilde{W}^{j}_{\tilde{g}}=U^{[j,j+N]*}_{\tilde{g}}W^{R_j*}_{\tilde{g}^{-1}} ,
\end{split}
\end{equation}
where $\tilde{g}$ is the element of $G$ in~\eqref{eq: sigma mixed}. Note that
\begin{align*}
    \pi(\widetilde{W}^{i*}_{\tilde{g}})\tilde{\Omega}_n &=U^{[i-N-1,i)}_{\tilde{g}} V^{L_i*}_{\tilde{g}} \tilde{\Omega}_n,\\
    \pi(\widetilde{W}^{j}_{\tilde{g}})\tilde{\Omega}_n &= U^{[j,j+N]*}_{\tilde{g}} V^{R_j}_{\tilde{g}} \tilde{\Omega}_n,
\end{align*}
by \eqref{eq: W^R Omega_n = V^R Omega_n}.

\begin{lemma}\label{lemma: finite correlation measured QCA}
    For all $i<j$ with $i-1,j\in K_n$,
    \begin{equation}\label{eq: finite correlation measured QCA}
    \left| \tilde \omega_n(\widetilde{W}^{i}_{\tilde{g}} \widetilde{W}^{j}_{\tilde{g}} ) -\tilde\omega_n(\widetilde{W}^{i}_{\tilde{g}} )\tilde \omega_n(\widetilde{W}^{j}_{\tilde{g}} )  \right| =1. 
\end{equation}
\end{lemma}

\begin{proof} 
We first show that $\widetilde{W}^{i}_{\tilde{g}}$ and $\widetilde{W}^{j}_{\tilde{g}}$ have a zero expectation value in the state $\tilde \omega_n$. First of all,
\begin{equation*}
    \tilde \omega_n(\widetilde{W}^j_{\tilde{g}} )=\langle\tilde\Omega_n,\pi(U^{[j,j+N]*}_{\tilde{g}})V^{R_j}_{\tilde{g}} \tilde\Omega_n\rangle \qquad 
\end{equation*}
As in Lemma~\ref{lemma: post meas SRE}, and since $\tilde h$ belongs to the centralizer of $G$, the state $\tilde\Omega_n$ is invariant under $V^{L_j}_{\tilde{h}}V^{R_j}_{\tilde{h}}$, and using non-degeneracy~\eqref{eq: sigma mixed}, we conclude that 
    \begin{align*}
    \tilde \omega_n(\widetilde{W}^j_{\tilde{g}})
    &=\overline{\tilde c_{\tilde h,n}^j}\langle\tilde \Omega_n,\pi(U^{[j,j+N]*}_{\tilde{g}})V^{R_j}_{\tilde{g}} V^{L_j}_{\tilde{h}}V^{R_j}_{\tilde{h}}\tilde\Omega_n\rangle\\
    &= \overline{\tilde c_{\tilde h,n}^j}\sigma_\omega(\tilde{g},\tilde{h})\langle\tilde\Omega_n,V^{L_j}_{\tilde{h}}V^{R_j}_{\tilde{h}}\pi(U^{[j,j+N]*}_{\tilde{g}})V^{R_j}_{\tilde{g}} \tilde\Omega_n\rangle\\
    &= \sigma_\omega(\tilde{g},\tilde{h})\tilde\omega_n(\widetilde{W}^j_{\tilde{g}}).
\end{align*}
Hence
\begin{equation*}
\tilde \omega_n(\widetilde{W}^j_{\tilde{g}})=0,
\end{equation*}
and similarly $\tilde \omega_n(\widetilde{W}^{i}_{\tilde{g}}  )=0$.

It remains to compute the expectation value of their product,
\begin{align*}
\tilde \omega_n( \widetilde{W}^{i}_{\tilde g}\widetilde{W}^j_{\tilde g} )
&= \langle \tilde \Omega_n, V_{\tilde g}^{L_i} \pi(U^{[i-N-1,i)*}_{\tilde{g}} 
U^{[j,j+N]*}_{\tilde{g}}) V_{\tilde g}^{R_j} \tilde\Omega_n \rangle.
\end{align*}
Shifting $V_{\tilde g}^{R_j}$ to $V_{\tilde g}^{R_i}$ using \eqref{eq: from V^R_i to V^R_j} and using the commutation of $V_{\tilde g}^{L_i}$ with every $U^{(k)*}_{\tilde g}$, we obtain
\begin{align*}
\tilde \omega_n( \widetilde{W}^i_{\tilde g}\widetilde{W}^j_{\tilde g} )
&{=} e^{i \gamma({\tilde g},i,j)} \langle \tilde\Omega_n, \pi(
U^{[i-N-1,j+N]*}_{\tilde{g}}
) V_{\tilde g}^{L_i} V_{\tilde g}^{R_i} \tilde\Omega_n \rangle\\
& = e^{i \gamma({\tilde g},i,j)} \langle \tilde\Omega_n, \pi( 
\bigotimes_{\substack{l\in K_n \\ i-1\leq l \leq j}} \widetilde{U}_{{\tilde g}}^{(l)*} 
) \tilde\Omega_n \rangle.
\end{align*}
Now, by  \eqref{eq: U P block}, we have that
\[
\pi ( \bigotimes_{\substack{l\in K_n \\ i-1\leq l \leq j}} \widetilde{U}_{\tilde g}^{(l)*}  ) \tilde\Omega_n 
=  \prod_{\substack{l\in K_n \\ i-1\leq l \leq j}} \chi_{q_l}(\tilde g)  \tilde\Omega_n,
\]
and since the characters are $U(1)$-valued, $\vert \tilde \omega_n (\widetilde{W}^i_{\tilde g}\widetilde{W}^j_{\tilde g}  ) \vert = 1$ as we had set to prove.
\end{proof}

\subsection{Proof of Theorem~\ref{theorem: long range correlations with QCA}}

The set of states on $\mathcal{A}$ is compact with respect to the weak-$*$ topology. Therefore, there exists a subsequence $(\tilde\omega_{n_m})_{m\in \mathbb{N}}$ and a state $\nu$ such that for all $A\in\mathcal{A}$, $\tilde\omega_{n_m}(A)\rightarrow \nu(A)$. Thanks to Lemma~\ref{lemma: finite correlation measured QCA}, we show that the state $\nu$ has arbitrarily long correlations. 
\begin{proof}
Let $r>0$, and let $m_0$ such that there is $i<j\in K_{n_{m_0}}$ and $|j-i|>r$. Then, for all $m\geq m_0$, $i,j\in K_{n_m}$.   
    
    \begin{equation*}
        \left| \tilde\omega_{n_m}(\widetilde{W}^i_{\tilde{g}} \widetilde{W}^j_{\tilde{g}} ) -\tilde\omega_{n_m}(\widetilde{W}^i_{\tilde{g}} )\tilde\omega_{n_m}(\widetilde{W}^j_{\tilde{g}} )  \right| =1.
    \end{equation*}
    Since the operators $\widetilde{W}^{i}_{\tilde{g}}$ and $\widetilde{W}^{j}_{\tilde{g}}$ do not depend on $m$ and are almost-localized, the weak convergence yields that
    \begin{equation}\label{eq:nu clustering}
        \left| \nu(\widetilde{W}^i_{\tilde{g}} \widetilde{W}^j_{\tilde{g}} ) -\nu(\widetilde{W}^i_{\tilde{g}} )\nu(\widetilde{W}^j_{\tilde{g}} )  \right| =1, \quad \forall i<j. 
    \end{equation}
    This is in contradiction with Lemma~\ref{lemma: decay of correlation}, thus the state $\nu$ is LRE. 
\end{proof}

\subsection{On the purity of the limiting state}\label{Section: nu is mixed}

To conclude, we recall the fact that pure states on the infinite chain must be clustering. In particular, (\ref{eq:nu clustering}) implies that, although the states $\tilde \omega_n$ that have seen a finite number of measurements are pure, their limit $\nu$ is not pure. 
\begin{lemma}\label{lem:pure-clustering}
Let $\nu$ be a pure state on the quasi-local algebra $\mathcal{A}$. Let $(B_m)_{m\in\mathbb{N}}\subset\mathcal{A}$ be a sequence with $\sup_m\Vert B_m\Vert<\infty$ and such that for every strictly local $A\in\mathcal{A}_{\rm loc}$ there is $m_A\in\mathbb{N}$ with
\begin{equation*}
    [B_m, A] = 0, \qquad \forall m\geq m_A.
\end{equation*}
Then
\begin{equation*}
    \lim_{m\to\infty}\big( \nu(AB_m) - \nu(A)\nu(B_m) \big) = 0, \qquad \forall A\in\mathcal{A}.
\end{equation*}
\end{lemma}

\begin{proof}
Let $(\mathcal{H}_\nu, \pi_\nu, \Omega_\nu)$ be a GNS triple for $\nu$. Since $\sup_m\Vert \pi(B_m)\Vert<\infty$ and the weak-* and weak operator topology agree on bounded sets, Banach-Alagoglu implies that $(\pi(B_m))_{m\in\mathbb{N}}$ has a weakly convergent subsequence $\pi(B_{m_k})\to B\in\mathcal{B}(\mathcal{H})$. For any $A\in\mathcal{A}_{\rm loc}$ and any $\xi,\eta\in\mathcal{H}_\nu$,
\begin{equation*}
    \langle \xi, [B,\pi_\nu(A)]\eta\rangle = \lim_{k\to\infty} \langle \xi, [\pi(B_{m_k}),\pi(A)]\eta\rangle = 0,
\end{equation*}
because the commutator vanishes for $m_k\geq m_{A}$. For $A\in\mathcal{A}$, $\Vert[B,\pi(A)]\Vert \leq \Vert[B,\pi(A_k)]\Vert + 2\Vert B\Vert\,\Vert A-A_k\Vert$, namely $[B,\pi(A)] = 0$. Hence $B\in\pi(\mathcal{A})'$. Since $\nu$ is pure, $\pi$ is irreducible and so $B=\lambda\mathbb{I}$ for some $\lambda\in\mathbb{C}$.

With this,
\begin{equation*}
    \begin{cases}
        &\nu(AB_{m_k}) = \langle\Omega_\nu, \pi_\nu(A)\pi_\nu(B_{m_k})\Omega\rangle \to \lambda\,\nu(A) \\
        &\nu(B_{m_k}) = \langle\Omega_\nu, \pi(B_{m_k})\Omega_\nu\rangle \to \lambda
    \end{cases}
\end{equation*}
and therefore $\nu(AB_{m_k})-\nu(A)\nu(B_{m_k})\to 0$. It remains to note that this holds for any limiting point of $(\pi(B_m))_{m\in\mathbb{N}}$ to conclude.
\end{proof}

\noindent Note that while any pure state is clustering, an SRE state exhibits fast decay of correlations, which is a much stronger condition.

\bigskip\bigskip
\noindent \textbf{Acknowledgments.} GT and SB would like to thank Robert Raussendorf and Andrew Potter for many inspiring discussions on measurements and entanglement. This work was supported by NSERC and the European Commission under the Grant \emph{Foundations of Quantum Computational Advantage}. SB and MR acknowledge further financial support from NSERC of Canada. 

\bigskip
\noindent \textbf{Conflict of Interest.} The authors declare no conflict of interest. 

\noindent \textbf{Data availability.} No data were used for the research described in the article.

\bibliographystyle{unsrt}
\bibliography{biblio}
\end{document}